\documentclass[a4paper, twoside,10pt]{article}

\usepackage{amsmath}
\usepackage{amssymb}
\usepackage{amsfonts}
\usepackage{graphicx}
\usepackage{subfigure}
\usepackage{eucal}
\usepackage{amsthm}
\usepackage{verbatim}
\usepackage{epsfig}
\usepackage{hyperref}

\def\Rc{\mathbb{R}}
\def\Cc{\mathbb{C}}

\newtheorem{thm}{Theorem}

\newtheorem{cor}{Corollary}

\newtheorem{remark}{Remark}

\begin{document}

\title{On decoding of digital data sent over a noisy MIMO channel
}
\author{Andrei Osipov\footnote{Yale University, 
51 Prospect st, New Haven, CT 06511.
Email: andrei.osipov@yale.edu.
}}
\maketitle

\begin{abstract}
The transmission of digital data is one of the principal
tasks in modern wireless communication.
Classically, the communication channel consists of 
one transmitter and one receiver; however,
due to the constantly increasing demand in higher transmission rates,
the popularity of using
several receivers and transmitters 
has been rapidly growing.

In this paper, we 
combine a number of fairly standard techniques
from numerical linear algebra and probability
to
develop several (apparently novel)
randomized schemes for
the decoding of digital messages sent over a noisy
multivariate Gaussian channel.

We use a popular mathematical model for such channels
to illustrate the performance of our schemes via 
numerical experiments.

\end{abstract}



\noindent
{\bf Math subject classification:} 
{65C99, 90C27, 94A12, 94B35}


\section{Introduction}
\label{sec_intro}

The importance of wireless communication can hardly be overestimated;
our life nowadays is unimaginable without it.
For obvious reasons, 
the transmission of digital (as opposed to analog) data
is one of the principal tasks in wireless communication.
In the case
when there is one transmitter and one receiver,
the subject has been extensively studied for decades.
More recently, however, 
as more and more data are collected
and
need to be rapidly transmitted,
there has been an increasing demand in
ways
to boost communication performance;
as a result, much research is being done on the
use of multiple receivers/transmitters.
This topic (generally referred to as
multiple-input and multiple-output, or MIMO) 
is now at the frontier of the research
in modern communications
(see e.g. \cite{Alamouti}, 
\cite{Erez}, 
\cite{Foschini}, 
\cite{Hochwald}, 
\cite{Jalden}, 
\cite{Koshy}, 
\cite{Koshy2}, 
\cite{Romero}, 
\cite{Telatar}).  

The transmission of digital data over a MIMO channel is typically
described by the following model
(see e.g. \cite{Hochwald}, \cite{Telatar}).
Suppose that $m>0$ is an integer power of 2, and that 
$C$ is a collection of $m$ points
in the complex plane
($C$ is usually referred to as a "constellation"); each element
of $C$ corresponds to a unique binary word of length $\log_2(m)$.

Suppose also that 
there are $n$ transmitters and $n$ receivers,
and that $X = C^m$ is the collection of all possible
vectors in $\Cc^m$ whose coordinates belong to $C$ 
(obviously, $X$ contains $m^n$ vectors). Suppose, in addition,
that $\sigma>0$ is a real number, and that $H$ is an $n$ by $n$
complex matrix (the "channel matrix"). For any $x$ in $X$,
we define the random $n-$dimensional complex vector $y(x)$
via the formula
\begin{align}
y(x) = H \cdot x + \sigma \cdot \left(z_1 + i \cdot z_2\right),
\label{eq_yx}
\end{align}
where $z_1,z_2$ are independent standard normal random vectors,
and $i=\sqrt{-1}$.
Here $x$ represents the transmitted message (that 
encodes $n \cdot \log_2(m)$ binary bits),
$H \cdot x$ represents the received message in the absence of noise,
and $y(x)$ represents the received message corrupted by Gaussian noise
of component-wise standard deviation $\sigma$.

In this model, the decoding problem can be formulated
as follows: suppose that $x_{\text{true}}$ is the
transmitted message, and that $y$ is the received message
(i.e. the observed value of $y(x_{\text{true}})$).
Under the assumption that $H$ and $\sigma$ are known,
one needs to find $x_{\text{best}}$ in $X$ such that
\begin{align}
\| H \cdot x_{\text{best}} - y \| \leq 
\| H \cdot x - y \|,
\label{eq_x_best_intro}
\end{align}
for any $x$ in $X$. In other words,
$x_{\text{best}}$ is the maximum likelihood estimate of
$x_{\text{true}}$.
 


In modern applications, any practical decoding algorithm must be quite rapid.
For instance, the peak speed requirement for 4G mobile telecommunication
technology is 1 Gigabit
per second, and
the 100 Gb/s RF Backbone DARPA project aims at 
the rate of 100 Gigabits per second (using
optical wireless communication).
To complicate the matters, the channel matrix $H$
needs to be frequently reevaluated (several times per second).

In principle, the decoding problem can be solved by brute force,
i.e. by iteratively testing all $x$ in $X$.
However, $X$ contains $m^n$ vectors, which is a large number
even for small values of $m$ and $n$ (e.g. $m=n=8$);
this makes the brute force approach impractical.

Another approach is 
based on the observation that, due to \eqref{eq_x_best},
$H \cdot x_{\text{best}}$ is the nearest neighbor of $y$
in the collection 
\begin{align}
HX = \left\{ H \cdot x \; : \; x \in X \right\}
\end{align}
of $m^n$ complex vectors in $\Cc^n$; thus, 
a fast nearest
neighbors algorithm might be used to find $x_{\text{best}}$.
Such algorithms typically pre-process $HX$ to obtain
auxiliary data structures; then, a single $y$ can be decoded
reasonably fast (see e.g. \cite{RannPnas}).
However, 
in a typical modern application $HX$ contains between $10^7$ and
$10^{10}$ vectors,
which makes 
the initial pre-processing of $HX$ unaffordable in terms of both
memory requirements and CPU time
(see e.g. \cite{RannPnas}),
especially considering that this calculation has to be redone
every time that $H$ changes.
In other words, even the fastest {\it generic} nearest neighbor search 
in $HX$ might not be fast enough.

Yet another
approach is based on the observation that, as opposed to $HX$,
the collection $X$ does not depend on $H$; moreover, it
has a special structure that allows for fast nearest neighbor
searches {\it in} $X$. Thus,
one can compute $\tilde{x}$ via the formula
\begin{align}
\tilde{x} = H^{-1} \cdot y,
\label{eq_tildex}
\end{align}
and look for $x_{\text{best}}$ among several nearest neighbors
of $\tilde{x}$ within $X$.
Unfortunately, this way the baby gets often thrown out with the bath water:
$x_{\text{best}}$ will typically
{\it not} be among any reasonable number of nearest neighbors of $\tilde{x}$
unless $\sigma$ is significantly smaller than the minimal
singular value $\lambda_n$ of $H$.
Suppose, for example, that $n=8$, that $C$
the so-called 8-PSK constellation of size $m=8$
(equispaced points on a unit circle), that $H$ is a randomly generated
complex Gaussian matrix, and that 
$\sigma = 0.75$
(i.e. roughly twice the expected minimal singular
value of such a matrix). 
Then, $x_{\text{best}}$ will coincide with $x_{\text{true}}$ in about
half of all cases; yet, in more than 70\% of these cases $x_{\text{best}}$
will not be among even as many as 500 nearest neighbors of $\tilde{x}$
in $X$.


Some other decoding schemes and approaches can be found, for example,
in
\cite{Alamouti}, 
\cite{Hochwald}, 
\cite{Jalden}, 
\cite{Koshy}, 
\cite{Koshy2} 
(see also references therein).

In this project, we develop several randomized decoding schemes. 
Our schemes are based on the observation that,
in typical applications, the channel matrix $H$ is
not large (e.g. $n=8$ or $n=16$),
and thus one can afford to evaluate its singular value decomposition (SVD).
Then, we use the SVD of $H$ to find $x_{\text{best}}$ by
a rapid randomized test-and-trial procedure.
Our tentative algorithms
demonstrate
reasonably good performance in several test cases
(see Section~\ref{sec_numerical}),
even when $\sigma = O(\lambda_n)$.
The schemes
break down only when $\sigma$ is so large that
the estimate
$x_{\text{best}}$ defined via \eqref{eq_x_best}
differs from 
the transmitted message $x_{\text{true}}$ with high probability
(e.g. it is impossible to recover $x_{\text{true}}$ 
without additional assumptions).


This paper is organized as follows.
Section~\ref{sec_prel} contains mathematical preliminaries
to be used in the rest of the paper.
Section~\ref{sec_analytical} contains a number of related theoretical
results. In Section~\ref{sec_num_algo}, we
describe the decoding problem and present several
decoding schemes.
In Section~\ref{sec_numerical}, we illustrate the performance
of our decoding schemes via numerical experiments.
In Section~\ref{sec_future}, we present some conclusions
and outline possible directions of future research.

\section{Mathematical Preliminaries}
\label{sec_prel}
In this section, we introduce notation and summarize
several facts to be used in the rest of the paper.

\subsection{Probability}
\label{sec_probability}
In this section, we summarize some well known facts
from the probability theory. These facts
can be found, for example, in
\cite{Abramovitz},
\cite{Billingsley},
\cite{Feller},
\cite{Grimmett}.

Suppose that $x > 0$ is a real number. 
In agreement with the standard practice, we define
the gamma function by the formula
\begin{align}
\Gamma(x) = \int_0^{\infty} t^{x-1} e^{-t} \; dt.
\label{eq_def_gamma}
\end{align}

The one-dimensional standard Gaussian distribution
$N(0,1)$ with mean zero and standard deviation one
is defined by its probability density function (pdf)
\begin{align}
f_{N(0,1)}(t) = \frac{ 1 }{ \sqrt{2\pi} } e^{-t^2/2}, 
\quad -\infty < t < \infty.
\label{eq_def_normal_pdf}
\end{align}
Its cumulative distribution function (cdf) is given by the formula
\begin{align}
\Phi(x) = \frac{1}{\sqrt{2\pi}} \int_{-\infty}^{x} e^{-t^2/2} \; dt
        = 1 - \frac{1}{2} \cdot \text{erfc}\left(\frac{x}{\sqrt{2}}\right),
\label{eq_def_normal_cdf}
\end{align}
where $\text{erfc} : \Rc \to \Rc$ is the 
complementary error function
defined via the formula
\begin{align}
\text{erfc}(x) = \frac{2}{\sqrt{\pi}} \int_x^{\infty} e^{-t^2} \; dt,
\label{eq_def_erfc}
\end{align}
for all real $x$.

Suppose that ${d} > 0$ is a positive integer. We say that the random
vector $X$ has standard Gaussian $d$-dimensional distribution
$N(0_d, I_d)$,
if all of its coordinates are independent standard Gaussian random variables.

Suppose now that $X \sim N(0_d, I_d)$.
Then $\| X \|^2$ has distribution\footnote{
Chi-square with $d$ degrees of freedom.}
$\chi^2(d)$
with pdf
\begin{align}
f_{\chi^2(d)}(t) = 
\frac{ t^{d/2-1} \cdot e^{-t/2} }{ 2^{d/2} \cdot \Gamma(d/2) }, \quad
t > 0,
\label{eq_def_chi2_pdf}
\end{align}
where $\Gamma$ denotes the gamma function defined via
\eqref{eq_def_gamma} above.
In particular, if $n>0$ is a positive integer, then
\begin{align}
f_{\chi^2(2n)}(t) = 
\frac{ t^{n-1} \cdot e^{-t/2} }{ 2^n \cdot (n-1)! },
\label{eq_chi2_pdf_even}
\end{align}
for all $t > 0$, and the corresponding cdf $F_{\chi^2(2n)}$ 
admits the form
\begin{align}
F_{\chi^2(2n)}(x) = 
1 - e^{-x/2} \cdot \sum_{k=0}^{n-1} \frac{ x^k }{ 2^k \cdot k!},
\label{eq_chi2_cdf_even}
\end{align}
for all $x>0$.

The real-valued Gaussian random variable has a straightforward
generalization to the complex plane. Suppose that $X, Y \sim N(0,1)$
are i.i.d. standard normal variables. We say that the random
variable $Z$ defined via the formula
\begin{align}
Z = X + i \cdot Y
\label{eq_ncomplex}
\end{align}
is a standard normal complex variable, and we denote its distribution by
$N\Cc(0,1)$.\footnote{in some sources, the same
distribution might be denoted $N\Cc(0,2)$ or $N\Cc(0,2;0)$.}
Similarly, if $d>0$ is a positive integer, we say that the complex
$d-$dimensional random vector $Z_d$ has complex standard 
distribution $N\Cc(0_d,I_d)$
if all of its coordinates are i.i.d. standard normal complex variables.
Needless to say, in this case
\begin{align}
\| Z_d \|^2 \sim \chi^2( 2 \cdot d).
\label{eq_complex_chi2}
\end{align}

The following well known theorem states the fact that
Gaussian distribution is invariant under orthogonal transformations.
\begin{thm}
Suppose that $d>0$ is a positive integer, that $U$ is the $d$ by $d$
unitary matrix, and that the complex $d-$dimensional random vector
$Z_d$ has distribution $N\Cc(0_d,I_d)$. Then,
\begin{align}
U \cdot Z_d \sim N\Cc(0_d,I_d).
\label{eq_invariant}
\end{align}
\label{thm_invariant}
\end{thm}

The following theorem describes some additional properties
of complex normal random variables.
\begin{thm}
Suppose that $n>0$ is a positive integer, and that $a_1,\dots,a_n$
are complex numbers. Suppose also that $z_1,\dots,z_n$ are i.i.d.
standard normal complex random variables. Then,
\begin{align}
a_1 \cdot z_1 + \dots + a_n \cdot z_n \sim
N\Cc(0,1) \cdot \sqrt{ |a_1|^2 + \dots + |a_n|^2 }.
\label{eq_std}
\end{align}
\label{thm_std}
\end{thm}

\subsection{Linear Algebra}
\label{sec_linear_algebra}

The following widely known theorem can be found,
in a more general form, in most
standard books on linear algebra
(see, for example, \cite{Golub}, \cite{Strang}).
\begin{thm}
Suppose that $n>0$ is a positive integer, and that $H$ is an $n$ by $n$
complex matrix. Then, there exist
non-negative real numbers $\sigma_1 \geq \dots \geq \sigma_n \geq 0$
and $n$ by $n$ unitary matrices $U$ and $V$ such that
\begin{align}
H = U \cdot \Sigma \cdot V^{\ast},
\label{eq_svd}
\end{align}
where $V^{\ast}$ denotes the conjugate transpose of $V$, and
$\Sigma$ is the diagonal $n$ by $n$ matrix whose diagonal entries
are defined via the formula
\begin{align}
\Sigma_{i,i} = \sigma_i,
\label{eq_svd_sigma}
\end{align}
for every $i=1,\dots,n$. The columns $u_1,\dots,u_n$ of $U$ are referred
to as left-singular vectors, the columns $v_1,\dots,v_n$ of $V$
are referred to as right-singular vectors, and $\sigma_1,\dots,\sigma_n$
are called the singular values of $H$. The factorization
\eqref{eq_svd} is typically referred to as the singular value
decomposition (SVD) of $H$.
\label{thm_svd}
\end{thm}


\section{Analytical Apparatus}
\label{sec_analytical}
The purpose of this section is to provide
the analytical apparatus to be used in the
rest of the paper.

\begin{thm}
Suppose that $n>0$ is a positive integer, that $H$ is
an $n$ by $n$ complex regular matrix, and that 
the matrices $U,\Sigma,V$ constitute the SVD of $H$,
as in \eqref{eq_svd} of Theorem~\ref{thm_svd}. Suppose, in addition,
that $w \sim N\Cc(0_d,I_d)$ is the standard normal complex random vector
in $\Cc^d$. Then,
\begin{align}
H^{-1} \cdot w \sim 
\frac{z_1}{\sigma_1} \cdot v_1 + \dots + \frac{z_n}{\sigma_n} \cdot v_n,
\label{eq_csn}
\end{align}
where the real numbers $\sigma_1,\dots,\sigma_n$ are the singular values
of $H$, the $n-$dimensional complex vectors $v_1,\dots,v_n$
are the columns of $V$, and $z_1,\dots,z_n$ are i.i.d. standard complex
normal random variables, as in \eqref{eq_ncomplex}.
\label{thm_csn}
\end{thm}
\begin{proof}
Due to \eqref{eq_svd} in Theorem~\ref{thm_svd},
\begin{align}
H^{-1} = V \cdot \Sigma^{-1} \cdot U^{\ast},
\label{eq_csn_1}
\end{align}
where $\Sigma^{-1}$ is the diagonal matrix whose diagonal entries are
$\sigma_1^{-1}, \dots, \sigma_n^{-1}$.
We combine \eqref{eq_invariant} in Theorem~\ref{thm_invariant} with
\eqref{eq_csn_1} to obtain \eqref{eq_csn}. 
\end{proof}
\begin{cor}
Suppose, in addition to the hypothesis of Theorem~\ref{thm_csn},
that $1 \leq i \leq n$ is an integer. Then,
\begin{align}
\left( H^{-1} \cdot w \right)(i) \sim
z_1 \cdot \sqrt{ \sum_{k=1}^{n-1} \frac{ |v_k(i)|^2 }{\sigma_k^2} }
+ z_n \cdot \frac{ |v_n(i)| }{ \sigma_n },
\label{eq_RaDe1}
\end{align}
where, for any complex vector $v$ in $\Cc^n$, we denote by $v(i)$
its i$th$ coordinate.
\label{cor_RaDe1}
\end{cor}
\begin{proof}
We combine Theorem~\ref{thm_std} with \eqref{eq_csn_1} to obtain
\eqref{eq_RaDe1}.
\end{proof}

\section{Numerical Algorithms}
\label{sec_num_algo}
In this section, we describe several numerical algorithms
for the decoding of a digital signal sent over a noisy MIMO channel,
in the sense described in Section~\ref{sec_task} below.

\subsection{Principal Decoding Task}
\label{sec_task}
In this section, we provide a formal description
of the principal decoding task
used in the rest of this paper
(see e.g. 
\cite{Alamouti}, 
\cite{Hochwald}, 
\cite{Jalden}, 
\cite{Koshy}, 
\cite{Koshy2}). 

Suppose that $n>0$ and $m>0$ are integers,
that 
\begin{align}
C = \left\{ c_1, \dots, c_m \right\}
\label{eq_c_gen}
\end{align}
is a collection of $m$ points in the complex plane
(often referred to as "constellation" in the literature),
that $H$ is an $n$ by $n$ complex matrix, and that $\sigma>0$
is a real number. Suppose also that $X$ is the collection
of $m^n$ vectors in $\Cc^n$ defined via the formula
\begin{align}
X = \left\{ \left(x(1),\dots,x(n)\right)^T \; : \; x(1),\dots,x(n) \in C
   \right\},
\label{eq_bigx_def}
\end{align}
and that $x_{\text{true}}$ is a vector in $X$ (we think of $x_{\text{true}}$
as of the unknown transmitted message). Suppose also that,
for every $x$ in $X$,  
the random complex $n-$dimensional vector $y(x)$ is defined via the formula
\begin{align}
y(x) = H \cdot x + \sigma \cdot Z_n,
\label{eq_y_def}
\end{align}
where $Z_n \sim N\Cc(0_d,I_d)$ is the standard normal $d-$dimensional
complex random vector
(see Section~\ref{sec_probability} above).

{\bf Task.} Suppose that $x_{\text{true}}$ in $X$ is the transmitted message,
and that $y_{\text{obs}}$ in $\Cc^n$ is the observed value of 
$y(x_{\text{true}})$. Given $X,H,\sigma$ and  $y_{\text{obs}}$,
one needs to find $x_{\text{true}}$. More specifically, 
one needs to find $x_{\text{best}}$ in $X$ such that
\begin{align}
\| H \cdot x_{\text{best}} - y_{\text{obs}} \| \leq 
\| H \cdot x - y_{\text{obs}} \|,
\label{eq_x_best}
\end{align}
for any $x$ in $X$. In other words,
$x_{\text{best}}$ is the maximum likelihood estimate of
$x_{\text{true}}$.
\begin{remark}
Obviously, for any fixed $H$ and any $x_{\text{true}}$ in $X$, 
the probability that
$x_{\text{best}}$ defined via \eqref{eq_x_best}
is equal to $x_{\text{true}}$ depends on $\sigma$
(more specifically, this probability decreases as $\sigma$ increases).
In particular, for noise of sufficiently large 
coordinate-wise standard deviation $\sigma$,
the best likelihood estimate is unlikely to coincide with
the transmitted message.
\label{rem_best}
\end{remark}
In the view of Remark~\ref{rem_best}, we make the following observations.

{\bf 1.} If $x_{\text{best}}$ defined via \eqref{eq_x_best} coincides
with $x_{\text{true}}$, a numerical scheme for decoding $y_{\text{obs}}$
should be able to recover $x_{\text{true}}$ (the task is well defined).

{\bf 2.} If $x_{\text{best}}$ defined via \eqref{eq_x_best} is different
from $x_{\text{true}}$, then there is not enough information for
recovering $x_{\text{true}}$ (without additional assumptions).
In that case, it might still be desirable for a decoding numerical scheme
to recover $x_{\text{best}}$.

\subsection{Running Example}
\label{sec_running}
For the sake of concreteness, we introduce the following example
(by specifying typical values of parameters of the model 
from Section~\ref{sec_task}),
to be used in the rest of this paper.

{\bf 1.} The dimensionality $n$ of the channel matrix 
(the number of transmitters and receivers):
\begin{align}
n = 6, 7, 8.
\label{eq_running_n}
\end{align}
While the value $n=8$ seems to be typical in many application,
the schemes should handle values up to $n=16$ or even $n=32$.

{\bf 2.} The number $m$ of constellation points:
\begin{align}
m = 8.
\label{eq_running_m}
\end{align}
While constellations of size $m=8$ are frequently used in applications,
many popular constellations contain $m=16$ or $m=32$ points, and these
values should be kept in mind.

{\bf 3.} The constellation $C$:
\begin{align}
C = \left\{ e^{i \cdot \pi \cdot k/4} \; : \; 0 \leq k < 8 \right\}
\label{eq_running_c}
\end{align}
(the so-called 8-PSK constellation). In other words, $C$ consists
of $m=8$ points equally distributed on the unit circle.

{\bf 4.} The channel matrix $H$: in agreement with a common practice,
we assume that the $n^2$ entries of $H$ have been independently
drawn from the standard normal complex distribution $N\Cc(0,1)$. 
In other words,
\begin{align}
H(i,j) \sim N\Cc(0,1),
\label{eq_running_h}
\end{align}
for every $i,j=1,\dots,n$.

In the rest of this section, we describe several decoding schemes.
Some of them are well known; others are apparently novel.
All of these schemes receive the observed value $y_{\text{obs}}$
of $y(x_{\text{true}})$ as an input and look for $x_{\text{best}}$
defined via \eqref{eq_x_best}.

\subsection{Brute force algorithm}
\label{sec_brute_force}

The brute force algorithm locates $x_{\text{best}}$ by
searching through all of $X$. In other words, it consists of 
the following steps:

{\bf Step 1.} For every $x$ in $X$, evaluate the real number
$d(x)$ via the formula
\begin{align}
d(x) = \| H \cdot x - y_{\text{obs}} \|.
\label{eq_dx_brute}
\end{align}

{\bf Step 2.} Define $x_{\text{best}}$ via finding the minimum
among all $d(x)$, i.e.
\begin{align}
d\left(x_{\text{best}}\right)  \leq d(x),
\label{eq_x_best_brute}
\end{align}
for all $x$ in $X$.
Obviously, \eqref{eq_x_best_brute} is equivalent to \eqref{eq_x_best}.

{\bf Memory requirements.} While the collection $X$ is typically quite large
(it contains $m^n$ vectors -- see \eqref{eq_bigx_def}),
all the vectors in $X$ can be iteratively computed one by one
in a straightforward way, thus obviating the need to pre-compute
$X$ and store it in memory (as far as the brute force algorithm
is concerned). Hence, the memory requirements of the brute force algorithm
are minimal in the sense that it requires only
\begin{align}
M_{\text{brute}} = O\left(n^2 + m\cdot n\right)
\label{eq_memory_brute}
\end{align}
memory words. In other words, practically speaking, the brute force
algorithm does not require any memory to speak of 
(as $m,n$ are typically very small: see Section~\ref{sec_running}).

{\bf Cost.} The cost of the brute force algorithm
is
\begin{align}
C_{\text{brute}} = O\left( |X| \cdot n^2 \right)
                 =  O\left( m^n \cdot n^2 \right)
\label{eq_cost_brute}
\end{align}
operations.

{\bf Success rate.} The brute force algorithm always locates 
$x_{\text{best}}$ and, in this sense, is optimal.

\begin{remark}
Even for as small values of $m$ and $n$ as $m=n=8$ (see 
Section~\ref{sec_running}), the brute force algorithm requires
about $8^{10} \approx 10^9$ operations. Since in a typical application
one needs to decode thousands or even millions messages per second,
its high cost deems the brute force decoding algorithm
impractical.
\label{rem_impractical} 
\end{remark}

{\bf Conclusion.} While the brute force algorithm is useless
in practical applications, it is a reasonable (albeit slow)
testing tool, due to its 100\% success rate.

We conclude by summarizing the principal input and output parameters
of the brute force algorithm described above.

{\bf Calling sequence:}
\begin{align}
\text{brute}(y_{\text{obs}}; x_{\text{best}}).
\label{eq_brute_calling}
\end{align}

{\bf Input parameters:}

-- $y_{\text{obs}}$ in $\Cc^n$: 

\quad the received message
(a noisy observation of $y(x_{\text{true}})$ defined via \eqref{eq_y_def}). \\

{\bf Output parameters:}

-- $x_{\text{best}}$ in $X$: 

\quad the maximal likelihood estimate of $x_{\text{true}}$
(see \eqref{eq_x_best}).

\subsection{Nearest Neighbors in $H \cdot X$}
\label{sec_nn_hx}
Obviously, the vector $H \cdot x_{\text{best}}$ 
(see \eqref{eq_x_best})
is simply the nearest neighbor of $y_{\text{obs}}$ within the collection
$HX$ of vectors in $\Cc^n$ defined via the formula
\begin{align}
H X = \left\{ H \cdot x \; : \; x \in X \right \}.
\label{eq_hx_def}
\end{align}
This observation suggests using a fast nearest neighbors algorithm
for computing $x_{\text{best}}$,
such as, for example, 
the randomized approximate nearest neighbors algorithm
(RANN) described in \cite{RannAcha}, \cite{RannPnas}
(needless to say, the straightforward nearest neighbors search
is simply the algorithm from Section~\ref{sec_brute_force}).
Obviously, a nearest neighbor search in $\Cc^n$ is
equivalent to that in $\Rc^{2n}$.

RANN consists of two steps briefly described below
(the reader is referred to \cite{RannAcha} for a more 
detailed description, and to \cite{Arya} for
another nearest neighbors algorithm):

{\bf Preprocessing (depends on $H$ but not on $y_{\text{obs}}$).}
One constructs a tree-like structure on $HX$, recursively subdividing
the points in $HX$ along each of $2n$ randomly chosen real axes.
The resulting "tree of boxes" consists of $2^{2n}$ boxes, each containing
$(m/4)^n$ points on average. Typically, several random subdivisions
are constructed (one can also further refine
each tree by subdividing boxes that contain too many points).

{\bf Query for nearest neighbors of $y_{\text{obs}}$.}
Once the auxiliary tree of boxes on $HX$ has been constructed, 
the search for nearest neighbors of any given $y_{\text{obs}}$
is done as follows. First, one locates the box that $y_{\text{obs}}$
belongs to; then, the points in this box (and, possibly, in
several boxes nearby) are inspected one by one.

{\bf Success rate.} The success rate depends on such parameters
as the number of trees used,
their internal structure (e.g. the
number of subdivisions in each tree, average number of points in a box), 
etc.
In particular,
there is a trade-off between probability
of locating
the nearest neighbor of $y_{\text{obs}}$ and required CPU time.
However, in this environment one can typically
achieve success rate of $90\%$ or higher 
without increasing CPU time too much; to that end, the parameters
of the nearest neighbor search are best tuned empirically.

{\bf Memory requirements.} Storing the data structures constructed
in the pre-processing step requires 
\begin{align}
M_{\text{preprocess}} = O\left( m^n \right)
\label{eq_memory_preprocess}
\end{align}
memory words per a single tree. In particular, even for such
small values of $m$ and $n$ as $m=n=8$ (see Section~\ref{sec_running}),
the memory requirements can easily exceed several Gigabytes.

{\bf Cost.} 
The cost of a single query for the nearest neighbor of $y_{\text{obs}}$
is proportional to the depth of the tree and the number
of points in a box; in particular, it will be of the order
\begin{align}
C_{\text{query}} = O\left( n^2 \cdot \log_2(m) \right)
\label{eq_cost_query}
\end{align}
(see e.g. \cite{RannAcha}). In other words, running time of a single
query is likely to be reasonable (compare to \eqref{eq_cost_brute},
and see also Section~\ref{sec_running}).
However,
the cost of the pre-processing step will be at least
\begin{align}
C_{\text{preproces}} = O\left( m^n \right)
\label{eq_cost_preprocess}
\end{align}
operations. While this step does not need to be redone for each
new $y_{\text{obs}}$, one still need to pre-process the data
from scratch every time the channel matrix $H$ changes.
In particular, in any application where $H$ is frequently
re-evaluated, this creates an additional major obstacle to
using such algorithms.

{\bf Conclusion.} We strongly suspect that {\it any}
generic nearest neighbors algorithm (used as a black box)
will be impractical for this decoding problem, 
in terms of either memory requirements
or CPU time (or both), simply due to the sheer amount of points
in $HX$. If, in addition, the channel matrix $H$ changes frequently,
such approaches are likely to be completely unaffordable.

\subsection{Nearest Neighbors in $X$}
\label{sec_nn_x}
In Section~\ref{sec_nn_hx} above, we discussed nearest neighbors
searches within the collection $HX$ defined via \eqref{eq_hx_def}.
In comparison, for any vector $y \in \Cc^n$ and integer $1 \leq k < n$,
to find $k$ nearest neighbors of $y$ within $X$ is a much easier task,
for the following reasons:

{\bf 1.} Due to \eqref{eq_bigx_def}, for any $x$ in $X$,
\begin{align}
\| y - x\|^2 = |y(1)-x(1)|^2 + \dots + |y(n)-x(n)|^2, 
\label{eq_yx_dist}
\end{align}
where $x(1),\dots,x(n)$ are in $C$. Thus, the search for nearest neighbors
is reduced to $n$ one-dimensional searches. For example, to find
the first nearest neighbor of $y$ in $X$ one simply needs to minimize
$|y(j)-x(j)|$ separately for every $j=1,\dots,n$.

{\bf 2.} The collection $X$ enjoys various symmetries that can be employed
in nearest neighbor searches. For example, if $C$ is defined via
\eqref{eq_running_c}, and one has found the first $k$ nearest neighbors
of
\begin{align}
x_0 = \left(1, \dots, 1\right)^T
\label{eq_x_special}
\end{align}
within $X$, then, for any $x$ in $X$ and every $1 \leq j \leq k$,
the $j$th nearest neighbor of $x$ within $X$ is computed
from that of $x_0$ by coordinate-wise rotation (in $n$ operations).
In other words, if $x_j$ is the $j$th nearest neighbor of $x_0$,
then the $j$th nearest neighbor $\hat{x}_j$ of 
$\hat{x}$ is defined via the formula
\begin{align}
\hat{x}_j = \left( x_j(1) \cdot \hat{x}(1), \dots,
x_j(n) \cdot \hat{x}(n) \right).
\label{eq_x_special_j}
\end{align}

{\bf 3.} As opposed to $HX$, 
the collection $X$ obviously does not depend on $H$,
and thus any pre-computation on $X$ needs to be done only once and for all.

These observations suggest the following algorithm for decoding 
$y_{\text{obs}}$:

{\bf Step 1.} Select a positive integer $1 \leq k < m^n$.

{\bf Step 2.} Compute the vector $\tilde{x}_{\text{obs}}$ in $\Cc^n$ via the formula
\begin{align}
\tilde{x}_{\text{obs}} = H^{-1} \cdot y_{\text{obs}}.
\label{eq_tilde_x_obs}
\end{align}

{\bf Step 3.} Find the $k$ nearest neighbors $x_1,\dots,x_k$ of
$\tilde{x}_{\text{obs}}$ within $X$.

{\bf Step 4.} For every $j=1,\dots,k$, evaluate $d_j$ via the formula
\begin{align}
d_j = \| H \cdot x_j - y_{\text{obs}} \|^2,
\label{eq_dj}
\end{align}
and find the index $j_0$ that corresponds to the minimal $d_j$.

{\bf Step 5.} Return $x_{\text{nn}(k)}$ defined via the formula
\begin{align}
x_{\text{nn}(k)} = x_{j_0}
\label{eq_x_nn}
\end{align}
(obviously, the subscript $k$ refers to the number of nearest neighbors
used to evaluate $x_{\text{nn}(k)}$).

The vector $x_{\text{nn}(k)}$ computed by this algorithm is
an approximation to $x_{\text{best}}$ defined via \eqref{eq_x_best}.

We summarize the principal input and output parameters
of the algorithm described in this section as follows.

{\bf Calling sequence:}
\begin{align}
\text{nnx}(y_{\text{obs}}; k; \tilde{x}_{\text{obs}}, x_{\text{nn}(k)}).
\label{eq_nnx_calling}
\end{align}

{\bf Input parameters:}

-- $y_{\text{obs}}$ in $\Cc^n$: 

\quad the received message
(a noisy observation of $y(x_{\text{true}})$ defined via \eqref{eq_y_def}). 

-- $k$ (a positive integer):

\quad the number of nearest neighbors to be used. \\

{\bf Output parameters:}

-- $\tilde{x}_{\text{obs}}$ in $X$: 

\quad the "naive" candidate for $x_{\text{true}}$ (see \eqref{eq_tilde_x_obs}).

-- $x_{\text{nn}(k)}$ in $X$:

\quad the candidate for $x_{\text{true}}$ 
found by nearest neighbors search in $X$
(see \eqref{eq_x_nn}). \\

{\bf Memory requirements.} The memory requirements for this
algorithm are minimal; one might only need to store $H$ and the list
of $k$ nearest neighbors of $\tilde{x}_{\text{obs}}$ in $X$, which requires
\begin{align}
M_{\text{nn}} = O(k \cdot n + n^2)
\label{eq_memory_nn}
\end{align}
memory words. Compared, for example, to 
\eqref{eq_memory_preprocess}, this is essentially negligible.

{\bf Cost.} It costs $O(n^3)$ operations to invert $H$; however,
if, for example, we have evaluated the SVD of $H$ beforehand,
the evaluation of $\tilde{x}_{\text{obs}}$ via \eqref{eq_tilde_x_obs} costs
only $O(n^2)$ operations. The evaluation of $x_1,\dots,x_k$
will typically require $O(n \cdot k)$ operations,
and, for every $j$, the evaluation of $d_j$ via \eqref{eq_dj}
requires $O(n^2)$ operations; thus, the total cost of the algorithm
is
\begin{align}
C_{\text{nn}} = O\left( k \cdot n^2 \right) 
\label{eq_cost_nn}
\end{align}
operations.

{\bf Success rate.} Unfortunately, unless $H$ is close to an orthogonal
matrix or $\sigma$ in \eqref{eq_y_def} is small compared to the smallest
singular value of $H$ (see Theorem~\ref{thm_svd}),
the success rate of this scheme will typically be poor
(loosely speaking, due to Theorem~\ref{thm_csn} -- in other words,
$H^{-1}$ magnifies the noise by different factors in different directions).
In the case of the model described in Section~\ref{sec_running},
the success rate of this scheme
is investigated empirically in Section~\ref{sec_numerical} below.

\subsection{Randomized Decoder 1 (RaDe1)}
\label{sec_RaDe1}
In this section, we described a randomized decoding scheme based
on Theorem~\ref{thm_csn} and Corollary~\ref{cor_RaDe1}.
This scheme will be referred to as "Randomized Decoder 1",
or "RaDe1".

\subsubsection{Preliminary discussion}
\label{sec_RaDe1_overview}
Suppose that $U,\Sigma,V$ constitute the SVD of $H$
(see Theorem~\ref{thm_svd}). 
Suppose also that the random vector $y(x_{\text{true}})$ 
is defined via \eqref{eq_y_def},
and that the random vector $\tilde{x}$ in $\Cc^n$
is defined via the formula
\begin{align}
\tilde{x} = H^{-1} \cdot y(x_{\text{true}}).
\label{eq_tilde_x}
\end{align}
Due to the combination of \eqref{eq_tilde_x} with
Theorems~\ref{thm_svd},~\ref{thm_csn},
\begin{align}
\tilde{x} \sim x_{\text{true}} =
\sum_{k=1}^n z_k \cdot \frac{\sigma}{\sigma_k} \cdot v_k,
\label{eq_tilde_x2}
\end{align}
where $v_1,\dots,v_n$ are the columns of $V$, $\sigma_1,\dots,\sigma_n$
are the singular values of $H$, and $z_1,\dots,z_n$ are i.i.d.
standard complex normal variables.

{\bf Observation 1.} The identity \eqref{eq_tilde_x2} admits the following
interpretation: while the distribution of $y(x_{\text{true}})$ is 
radially symmetric about $H \cdot x_{\text{true}}$, the distribution
of $\tilde{x}$ is not. More specifically, the "noise" in $\tilde{x}$
has the largest variance in the direction of $v_n$ and the
smallest variance in the direction of $v_1$. In other words,
$z_n$ has larger effect on how far $\tilde{x}$ is from $x_{\text{true}}$
that $z_1$ does.

Suppose now that, for every $j=1,\dots,n$ and every $k=1,\dots,n$,
we define the real numbers $s_{k}(j)$ and $S_k(j)$ via the formulae
\begin{align}
s_k(j) = \frac{ |v_k(j)|}{\sigma_k},
\label{eq_skj}
\end{align}
\begin{align}
S_k(j) = \sqrt{ \sum_{l=1}^k s_l^2(j) },
\label{eq_skj_2}
\end{align}
respectively.
It follows from the combination of \eqref{eq_skj}, \eqref{eq_skj_2} and
Corollary~\ref{cor_RaDe1} that
\begin{align}
\tilde{x}(j) \sim x_{\text{true}}(j) +
z \cdot \sigma \cdot S_{n-1}(j) + z_n \cdot \sigma \cdot s_n(j), 
\label{eq_xj}
\end{align}
for every $j=1,\dots,n$,
where $z,z_n$ are i.i.d. $N\Cc(0,1)$
(obviously, $\tilde{x}(1),\dots,\tilde{x}(n)$ are not independent
of each other). 

The coefficients $S_{n-1}(j)$ and $s_n(j)$ in \eqref{eq_xj}
are determined by $H$. In the case
when the entries of $H$ have been drawn independently
from $N\Cc(0,1)$ (see Section~\ref{sec_running}), we
can make the following observation.

{\bf Observation 2.} Suppose that $6 \leq n \leq 32$ is an integer,
that the entries of $H$ are i.i.d. $N\Cc(0,1)$, and that the real
random variable $r=r(H)$ is defined via the formula
\begin{align}
r(H) = \frac{s_n(j)}{S_{n-1}(j)},
\label{eq_rdef}
\end{align}
where, for any $H$,  
$1 \leq j \leq n$ is the integer such that $S_{n-1}(j) \leq S_{n-1}(i)$
for every $i = 1, \dots, n$. Then,
\begin{align}
0.25 \leq \mathbb{E}\left[ S_{n-1}(j) \right] \leq 0.3,
\label{eq_exp_snm}
\end{align}
and also
\begin{align}
4.5 \leq \mathbb{E}\left[ r(H) \right] \leq 6.
\label{eq_exp_r}
\end{align}
In other words, the coefficient in front of $z_n$ in \eqref{eq_xj} 
is about 5 times larger (on average) than that in front of $z$.

\subsubsection{RaDe1: informal description}
\label{sec_RaDe1_informal}

The discussion in Section~\ref{sec_RaDe1_overview}
above leads to the following decoding scheme
(based primarily on \eqref{eq_xj} and Observation 2 above).

We select $1 \leq j \leq n$ such that $S_{n-1}(j)$ in \eqref{eq_xj}
is minimal. Then, for every $i=1,\dots,m$,
we assume that 
\begin{align}
x_{\text{true}}(j) = c_i,
\label{eq_xj_ass}
\end{align}
and proceed as follows. Under the assumption \eqref{eq_xj_ass},
the difference $\tilde{x}(j)-x_{\text{true}}(j)$ is a sum
of two independent complex normal variables with different variances
(see \eqref{eq_xj}). We sample one of them; more specifically,
we draw the complex number $\tilde{z}$ from the distribution $N\Cc(0,1)$.
The assumptions \eqref{eq_xj_ass} and $z=\tilde{z}$ determine
the value of $z_n$ in \eqref{eq_xj}.
However, we observe that $z_n$ is "responsible" for the largest component of
the noise in $\tilde{x}$, due to \eqref{eq_tilde_x2}. Thus, we can reduce
the noise by subtracting $z_n \cdot \sigma \cdot \sigma_n^{-1} \cdot v_n$
from the observed value of $\tilde{x}$. Next, we look for the nearest
neighbor $x_i$ of the "improved" $\tilde{x}$ in $X$
(a simple task; see Section~\ref{sec_nn_x}). We observe that,
provided that the assumption \eqref{eq_xj_ass} is correct and
the noise in $\tilde{x}$ has been reduced, $x_i$ is more likely
to be equal to $x_{\text{true}}$ than if we were to use
the nearest neighbor search in $X$ for the original $\tilde{x}$.

To validate the assumption \eqref{eq_xj_ass}, we make several observations.

{\bf 1.} For each $i=1,\dots,m$, the assumption \eqref{eq_xj_ass}
yields $x_i$. We can simply select the best $x_i$ out of $x_1,\dots,x_m$ 
(the one for which $\| H \cdot x_i - y_{\text{obs}} \|$ is
the smallest).

{\bf 2.} If $x_i = x_{\text{true}}$, then we can determine
the noise from \eqref{eq_y_def}; in particular, we can
compute the square of the Euclidean norm of the noise.
However, this quantity must be an observed value of a
$\chi^2(2n)$ random variable (multiplied by $\sigma^2$),
whose cdf is defined via \eqref{eq_chi2_cdf_even}; we use it
to estimate our "confidence" in the statement that $x_i = x_{\text{true}}$.

{\bf 3.} The value of $z_n$ must be an observed value of a standard normal
complex variable. In particular, $|z_n|^2 \sim \chi^2(2)$, and we
can reject the assumption \eqref{eq_xj_ass} if $|z_n|^2$ is too large.

Finally, we observe that the scheme described above is randomized.
The apparent downside is that, even when the assumption \eqref{eq_xj_ass}
is correct for a particular value of $i$, there is a finite probability
that $x_i \neq x_{\text{best}}$ (see \eqref{eq_x_best}). However,
the obvious advantage of randomization is that the scheme
can be iterated (and the probability of failure will decrease 
with the number of iterations; see, however, 
Remark~\ref{rem_RaDe1_stuck} in Section~\ref{sec_exp_3}).

\subsubsection{Basic RaDe1: detailed description}
\label{sec_RaDe1_detailed}

This section contains a detailed description of the decoding
scheme described in Section~\ref{sec_RaDe1_informal} above. \\

{\bf Precomputation.} Suppose that $H$ is the $n$ by $n$
channel matrix. \\

{\bf Step 1.} Evaluate the SVD (e.g.
the matrices $U,\Sigma,V$) of $H$ (see Theorem~\ref{thm_svd}).

{\bf Step 2.} For every $k=1,\dots,n$, evaluate 
$s_n(k)$ and $S_{n-1}(k)$
via \eqref{eq_skj}, \eqref{eq_skj_2}, respectively.

{\bf Step 3.} Find $1 \leq j \leq n$ such that $S_{n-1}(j) \leq S_{n-1}(k)$
for every $k=1,\dots,n$. \\

{\bf Decoding.} Suppose that $x_{\text{true}}$
in $X$ is the (unknown) transmitted message, 
that $\sigma > 0$ is the component-wise variation of the noise in 
\eqref{eq_y_def}, and 
that $y_{\text{obs}}$ is the observed
value of $y(x_{\text{true}})$ (see \eqref{eq_y_def}). \\

For every $i=1,\dots,m$ proceed as follows: \\

{\bf Step 1.} Sample the complex number $\hat{z} \sim N\Cc(0,1)$.

{\bf Step 2.} Evaluate $\tilde{x}_{\text{obs}}$ in $\Cc^n$ via 
\eqref{eq_tilde_x_obs}.

{\bf Step 3.} Evaluate $\hat{z}_n$ via the formula
\begin{align}
\hat{z}_n = 
\frac{ \tilde{x}_{\text{obs}}(j) - c_i - 
       \hat{z} \cdot \sigma \cdot S_{n-1}(j)}
{v_n(j)},
\label{eq_zn_hat}
\end{align}
where $v_n$ in $\Cc^n$ is the $n$th column of $V$.

{\bf Step 4.} Evaluate $\tilde{x}_i$ in $\Cc^n$ via the formula
\begin{align}
\tilde{x}_i = \tilde{x}_{\text{obs}} - \hat{z}_n \cdot v_n.
\label{eq_tilde_x_j}
\end{align}

{\bf Step 5.} Find the nearest neighbor $x_i$ of $\tilde{x}_i$ in $X$.

{\bf Step 6.} Evaluate $w_i$ in $\Cc^n$ via the formula
\begin{align}
w_i = y_{\text{obs}} - H \cdot x_i.
\label{eq_w_i}
\end{align}

{\bf Step 7.} Evaluate the real numbers $r_i$ and $\chi_i$ via the formulae
\begin{align}
r_i = \frac{ |w_i(1)|^2 + \dots + |w_i(n)|^2 }{\sigma^2}
\label{eq_r_i}
\end{align}
and
\begin{align}
\chi_i = 1 - F_{\chi^2(2n)}(r_i),
\label{eq_chi_i}
\end{align}
respectively,
where $F_{\chi^2(2n)}$ is the cdf of the $\chi^2_{2n}$-distribution
(see \eqref{eq_chi2_cdf_even}). \\

Thus, for each $i=1,\dots,m$, the scheme produces $x_i \in X$
and the real numbers $r_i$ and $\chi_i$. Suppose
that $i_0$ is the index of the smallest
$r_i$ (and thus largest $\chi_i$) among $r_1,\dots,r_n$. The algorithm
returns the vector $x_{\text{RaDe1}}$ in $X$ defined via the formula
\begin{align}
x_{\text{RaDe1}} = x_{i_0}.
\label{eq_x_RaDe1}
\end{align}

{\bf Observation.} Under the assumption that 
$x_{\text{true}} = x_i$, 
the vector $w_i$ defined via \eqref{eq_w_i} is an observed value
of the random vector having the distribution 
$\sigma \cdot N\Cc(0_d,I_d)$ (see \eqref{eq_y_def}). Consequently,
$r_i$ is an observed value from the $\chi^2_{2n}$ distribution,
and $\chi_i$ is the probability that a $\chi^2_{2n}$ random variable
attains values larger than $r_i$ (we refer to $\chi_i$ as
the "confidence").

As a conclusion, we summarize the principal input and output parameters
of the decoding scheme described in this section.

{\bf Calling sequence:}
\begin{align}
\text{RaDe1\_search}\left( y_{\text{obs}}; x_{\text{RaDe1}}, r, \chi\right).
\label{eq_RaDe1_calling}
\end{align}

{\bf Input parameters:}

-- $y_{\text{obs}}$ in $\Cc^n$: 

\quad the received message
(a noisy observation of $y(x_{\text{true}})$ defined via \eqref{eq_y_def}). \\

{\bf Output parameters:}

-- $x_{\text{RaDe1}}$ in $X$: 

\quad the candidate for the transmitted message (see \eqref{eq_x_RaDe1}).

-- $r$ (a positive real number):

\quad the squared norm of the noise 
     (provided that $x_{\text{true}} = x_{\text{RaDe1}}$), see \eqref{eq_r_i}.

-- $\chi$ (a real number between 0 and 1):

\quad the confidence in $x_{\text{RaDe1}}$ (see \eqref{eq_chi_i}).

\subsubsection{RaDe1: full algorithm}
\label{sec_RaDe1_full}

In this section, we describe a decoding algorithm
whose basic step is the 
scheme described above (see Sections~\ref{sec_RaDe1_informal},
\ref{sec_RaDe1_detailed}).

{\bf Additional parameters.}

-- \text{min\_RaDe1}: a positive integer
(minimal number of iterations)

-- $\text{max\_RaDe1} \geq \text{min\_RaDe1}$: a positive integer
(maximal number of iterations)

-- $\chi_{\text{thresh}}$: a real number between 0 and 1
(a confidence threshold)

{\bf Description.}

{\bf 1.} For each $1 \leq j < \text{min\_RaDe1}$, 
call $\text{RaDe1\_search}\left( y_{\text{obs}}; x_j, r_j, \chi_j\right)$.

{\bf 2.} Among all $(x_j,r_j,\chi_j)$, select the triplet
\begin{align}
\left(x_{\text{RaDe1}}, r_{\text{RaDe1}}, \chi_{\text{RaDe1}}\right) =
(x_i,r_i,\chi_i)
\label{eq_trip}
\end{align}
that corresponds to the smallest $r_i$.

{\bf 3.} For each $\text{min\_RaDe1} \leq j \leq \text{max\_RaDe1}$:

\quad {\bf 3a.} 
call $\text{RaDe1\_search}\left( y_{\text{obs}}; x_j, r_j, \chi_j\right)$.

\quad {\bf 3b.}
if $r_j < r_{\text{RaDe1}}$, set
\begin{align}
\left(x_{\text{RaDe1}}, r_{\text{RaDe1}}, \chi_{\text{RaDe1}}\right) =
(x_j,r_j,\chi_j).
\label{eq_trip2}
\end{align}

\quad {\bf 3c.} if $\chi_{\text{RaDe1}} > \chi_{\text{thresh}}$,
stop. \\

{\bf Comment.} The algorithm conducts at least min\_RaDe1 iterations
of the basic scheme. Then, if the confidence in the best guess
is high enough (compared to $\chi_{\text{thresh}}$), 
the algorithm stops. Otherwise, the basic scheme
is called iteratively until the confidence is high enough
{\it or} the total number of iterations reaches max\_RaDe1.

To conclude, we summarize the principal input and output parameters
of the algorithm described in this section.

{\bf Calling sequence:}
\begin{align}
\text{RaDe1\_all}\left( y_{\text{obs}}; 
\text{min\_RaDe1}, \text{max\_RaDe1}, \chi_{\text{thresh}};
x_{\text{RaDe1}}, r_{\text{RaDe1}}, \chi_{\text{RaDe1}}\right).
\label{eq_RaDe1_all_calling}
\end{align}

{\bf Input parameters:}

-- $y_{\text{obs}}$ in $\Cc^n$ (see Section~\ref{sec_RaDe1_detailed}). 

-- min\_RaDe1: see above.

-- max\_RaDe1: see above.

-- $0 < \chi_{\text{thresh}} < 1$: see above. \\

{\bf Output parameters:}

-- $x_{\text{RaDe1}}$ in $X$ (see \eqref{eq_trip}).

-- $r_{\text{RaDe1}}$ (a positive real number): see \eqref{eq_trip}.

-- $\chi_{\text{RaDe1}}$ (a real number between 0 and 1): see \eqref{eq_trip}.

\subsubsection{RaDe1: cost and memory requirements}
\label{sec_RaDe1_cost}

In this section, we describe the cost and memory requirements
of the basic decoding scheme described in Sections~\ref{sec_RaDe1_informal}, 
\ref{sec_RaDe1_detailed} above.

{\bf Memory requirements.} The memory requirements of the scheme
are or the order
\begin{align}
M_{\text{RaDe1}} = O\left( m \cdot n + n^2 \right)
\label{eq_memory_RaDe1}
\end{align}
memory words - in other words, absolutely minimal
(compare for example to
\eqref{eq_memory_preprocess}; see also Section~\ref{sec_running}
for typical values of $m$ and $n$).

{\bf Cost.} The precomputation step (see Section~\ref{sec_RaDe1_detailed})
requires $O(n^3)$ operations, and should not be re-done until
$H$ changes.

On the other hand, each decoding step requires
\begin{align}
C_{\text{RaDe1}} = O\left( m \cdot n^2 \right)
\label{eq_cost_RaDe1}
\end{align}
operations. Therefore, the full algorithm described in 
Section~\ref{sec_RaDe1_full} requires between
\begin{align}
C_{\text{RaDe1,best}} = O\left( \text{min\_RaDe1} \cdot m \cdot n^2 \right)
\label{eq_cost_best}
\end{align}
and
\begin{align}
C_{\text{RaDe1,worst}} = O\left( \text{max\_RaDe1} \cdot m \cdot n^2 \right)
\label{eq_cost_worst}
\end{align}
operations.

{\bf Success rate.} Obviously, the success rate of RaDe1
depends on various parameters of the problems (e.g. $H, \sigma, X$)
as well as on the parameters of the algorithm
(see Section~\ref{sec_RaDe1_full}). In addition, since the scheme
is randomized, for each set of conditions, there is a certain probability
of failure (that decreases with the number
of iterations of the basic step).
Obviously, there is a trade-off between the probability of success
and number of iterations.

The performance of the scheme is demonstrated via several experiments
in Section~\ref{sec_numerical}.

\subsection{Randomized Decoder 2 (RaDe2)}
\label{sec_RaDe2}

In this section, we describe yet another decoding scheme.
This scheme is closely related to RaDe1 from Section~\ref{sec_RaDe1},
and can be viewed as a generalization of the latter.
This scheme will be referred to as "Randomized Decoder 2",
or "RaDe2".

\subsubsection{Preliminary discussion}
\label{sec_RaDe2_overview}
The decoding scheme RaDe1 from Section~\ref{sec_RaDe1}
is based on the equation \eqref{eq_xj}. Loosely speaking,
one fixes the coordinate $j$, "guesses" the value of $x_{\text{true}}(j)$,
samples $z$ from $N\Cc(0,1)$, and determines the value
of $z_n$ from \eqref{eq_xj}. Then, the noise in $\tilde{x}$ is
reduced by subtracting its (estimated) largest component in the direction
of $v_n$ (see \eqref{eq_tilde_x2}).

In this section, we describe an algorithm that uses two 
(rather than one)
coordinates. The resulting scheme is more accurate than
RaDe1 (see Section~\ref{sec_RaDe1}); however, it is also more
computationally expensive. 

Along the lines of Section~\ref{sec_RaDe1_overview}, we make the following
observation. Suppose that $1 \leq j_1, j_2 \leq n$ correspond
to the two smallest values of $S_{n-2}(1), \dots, S_{n-2}(n)$ defined
via \eqref{eq_skj_2}. Due to the combination of 
\eqref{eq_tilde_x2},
\eqref{eq_skj},
\eqref{eq_skj_2},
\begin{align}
& \tilde{x}(j_1) \sim x_{\text{true}}(j_1) +
u \cdot \sigma \cdot S_{n-2}(j_1) + 
z_{n-1} \cdot \sigma \cdot s_{n-1}(j_1) +
z_n \cdot \sigma \cdot s_n(j_1), \\ 
& \tilde{x}(j_2) \sim x_{\text{true}}(j_2) +
w \cdot \sigma \cdot S_{n-2}(j_2) + 
z_{n-1} \cdot \sigma \cdot s_{n-1}(j_2) +
z_n \cdot \sigma \cdot s_n(j_2),
\label{eq_xj1}
\end{align}
where $u,w \sim N\Cc(0,1)$ (however, $u$ and $w$ are not independent
of each other), while $z_{n-1}, z_n \sim N\Cc(0,1)$
are independent of each other and of $u,w$.

Obviously, due to \eqref{eq_skj_2}, 
the coefficients of $u, w$ in \eqref{eq_xj1}
are even smaller than the coefficient of $z$ in \eqref{eq_xj}
(see Section~\ref{sec_RaDe1}, in particular Observation 2).

\subsubsection{RaDe2: informal description}
\label{sec_RaDe2_informal}

The discussion in Section~\ref{sec_RaDe2_overview}
above leads to the following decoding scheme
(somewhat similar to that from Section~\ref{sec_RaDe1}).

We select $1 \leq j_1,j_2 \leq n$ such that 
$S_{n-2}(j_1), S_{n-2}(j_2)$, respectively,
are the smallest among $S_{n-2}(1), \dots, S_{n-2}(n)$.
Then, for every {\it pair} of indices
$1 \leq i_1, i_2 \leq m$, we assume that
\begin{align}
& x_{\text{true}}(j_1) = c_{i_1}, \\
& x_{\text{true}}(j_2) = c_{i_2},
\label{eq_xj1_ass}
\end{align}
and proceed as follows. 
Under the assumption \eqref{eq_xj1_ass},
the difference $\tilde{x}(j)-x_{\text{true}}(j)$ for each $j=j_1,j_2$
is a sum
of three independent complex normal variables with different variances
(see \eqref{eq_xj1}). 
We sample $\hat{u},\hat{w}$ in \eqref{eq_xj1} by drawing i.i.d. 
$z_1,\dots,z_{n-2} \sim N\Cc(0,1)$ and using \eqref{eq_tilde_x2}.
Then, we determine the values of $z_{n-1}, z_n$ in
\eqref{eq_xj1} under the assumption \eqref{eq_xj1_ass} and
$u=\hat{u}, w=\hat{w}$ by solving the corresponding
two by two linear system.
We reduce the noise in $\tilde{x}$ by subtracting
$z_{n-1} \cdot \sigma \cdot \sigma_{n-1}^{-1} \cdot v_{n-1}$ and
$z_n \cdot \sigma \cdot \sigma_n^{-1} \cdot v_n$ from
the observed value of $\tilde{x}$ (see \eqref{eq_tilde_x2}).

Next, we look for the nearest
neighbor $x_{i_1,i_2}$ of the "improved" $\tilde{x}$ in $X$
(see Section~\ref{sec_nn_x}). We observe that,
provided that the assumption \eqref{eq_xj1_ass} is correct and
also the noise in $\tilde{x}$ has indeed been reduced, 
$x_{i_1,i_2}$ is more likely
to be equal to $x_{\text{true}}$ than if we were to use
the nearest neighbor search in $X$ for the original $\tilde{x}$.

To validate the assumption \eqref{eq_xj_ass}, we make several observations.

{\bf 1.} For each $1 \leq i_1,i_2 \leq m$, the assumption \eqref{eq_xj1_ass}
yields $x_{i_1,i_2}$. 
We can simply select the best one out of the $m^2$ possible 
$x_{i_1,i_2}$ (actually, we do not even need to examine all $m^2$
possibilities; see Section~\ref{sec_RaDe2_detailed} below).

{\bf 2.} If $x_{i_1,i_2} = x_{\text{true}}$, then we can determine
the noise from \eqref{eq_y_def}; in particular, we can
compute the square of the Euclidean norm of the noise.
However, this quantity must be an observed value of a
$\chi^2(2n)$ random variable (multiplied by $\sigma^2$),
whose cdf is defined via \eqref{eq_chi2_cdf_even}; we use it
to estimate our "confidence" in the statement that 
$x_{i_1,i_2} = x_{\text{true}}$.

{\bf 3.} The value of $|z_{n-1}|^2 + |z_n|^2$ 
must be an observed value of a $\chi^2(4)$ random variable, and thus we
can reject the assumption \eqref{eq_xj1_ass} 
{\it before} evaluating $x_{i_1,i_2}$
if $|z_n|^2$ is too large.

\subsubsection{Basic RaDe2: detailed description}
\label{sec_RaDe2_detailed}

This section contains a detailed description of the decoding
scheme described in Section~\ref{sec_RaDe2_informal} above. 

Compared to the scheme described in Section~\ref{sec_RaDe1_detailed},
the procedure below accepts an additional input parameter:
namely, the "early exit confidence" $\chi_{\text{stop}}$
(a real number between 0 and 1: see \eqref{eq_chi_stop} below). \\

{\bf Precomputation.} Suppose that $H$ is the $n$ by $n$
channel matrix. \\

{\bf Step 1.} Evaluate the SVD (e.g.
the matrices $U,\Sigma,V$) of $H$ (see Theorem~\ref{thm_svd}).

{\bf Step 2.} For every $k=1,\dots,n$, evaluate 
$s_n(k), s_{n-1}(k)$ and $S_{n-2}(k)$
via \eqref{eq_skj}, \eqref{eq_skj_2}, respectively.

{\bf Step 3.} Find $1 \leq j_1,j_2 \leq n$ such that 
$S_{n-2}(j_1), S_{n-2}(j_2)$ are the two smallest values among 
$S_{n-2}(1), \dots, S_{n-2}(n)$. \\

{\bf Decoding.} Suppose that $x_{\text{true}}$
in $X$ is the (unknown) transmitted message, 
that $\sigma > 0$ is the component-wise variation of the noise in 
\eqref{eq_y_def}, and 
that $y_{\text{obs}}$ is the observed
value of $y(x_{\text{true}})$ (see \eqref{eq_y_def}). \\

For every $i_1=1,\dots,m$ and $i_2 = 1,\dots, m$, proceed as follows: \\

{\bf Step 1.} Sample the $n-2$ complex numbers 
$\hat{z_1}, \dots, \hat{z_{n-2}}$ independently from $N\Cc(0,1)$.

{\bf Step 2.} Evaluate $\hat{u}, \hat{w}$ via the formulae
\begin{align}
& \hat{u} = 
\sum_{k=1}^{n-2} \hat{z}_k \cdot \frac{\sigma}{\sigma_k} \cdot v_k(j_1), \\
& \hat{w} = 
\sum_{k=1}^{n-2} \hat{z}_k \cdot \frac{\sigma}{\sigma_k} \cdot v_k(j_2),
\label{eq_tilde_uw}
\end{align}
where $v_k$ in $\Cc^n$ is the $k$th column of $V$ for every $k=1,\dots,n$.

{\bf Step 3.} Evaluate $\tilde{x}_{\text{obs}}$ in $\Cc^n$ via 
\eqref{eq_tilde_x_obs}.

{\bf Step 4.} Evaluate $\hat{z}_{n-1}, \hat{z}_n$ via 
solving the two by two linear system
\begin{align}
& z_{n-1} \cdot \frac{\sigma \cdot v_{n-1}(j_1)}{\sigma_{n-1}} 
+
z_n \cdot \frac{\sigma \cdot v_n(j_1)}{\sigma_n} 
=
\tilde{x}_{\text{obs}}(j_1) - c_{i_1} - \hat{u}, \\
& z_{n-1} \cdot \frac{\sigma \cdot v_{n-1}(j_2)}{\sigma_{n-1}} 
+
z_n \cdot \frac{\sigma \cdot v_n(j_2)}{\sigma_n} 
=
\tilde{x}_{\text{obs}}(j_2) - c_{i_2} - \hat{w},
\label{eq_znmn}
\end{align}
in the unknowns $z_{n-1}, z_n$.

{\bf Step 5.} Evaluate the real numbers $r$ and $\chi$ via the formulae
\begin{align}
& r = |\hat{z}_{n-1}|^2 + |\hat{z}_n|^2, \\
& \chi = 1-F_{\chi^2(4)}(r), 
\label{eq_rchi}
\end{align}
respectively, where $F_{\chi^2(4)}$ is the cdf of the distribution
$\chi^2(4)$ defined via the formula \eqref{eq_chi2_cdf_even}.

{\bf Step 6.} If 
\begin{align}
\chi < \chi_{\text{stop}}, 
\label{eq_chi_stop}
\end{align}
skip the rest of
the steps for these values of $i_1,i_2$.

{\bf Comment.} Roughly speaking, 
if $\chi$ is too small, it means that $\hat{z}_{n-1}$
and $\hat{z}_n$ are unlikely to be observed values
of two i.i.d. $N\Cc(0,1)$ random variables; in other words,
the assumption \eqref{eq_xj1_ass} is likely to be wrong
for these values of $i_1,i_2$.

{\bf Step 7.} Evaluate $\tilde{x}_{i_1,i_2}$ in $\Cc^n$ via the formula
\begin{align}
\tilde{x}_{i_1,i_2} = \tilde{x}_{\text{obs}} 
- \hat{z}_{n-1} \cdot \frac{\sigma}{\sigma_{n-1}} \cdot v_{n-1}
- \hat{z}_n \cdot \frac{\sigma}{\sigma_n} \cdot v_n
\label{eq_tilde_x_jj}
\end{align}

{\bf Step 8.} Find the nearest neighbor $x_{i_1,i_2}$ of 
$\tilde{x}_{i_1,i_2}$ in $X$.

{\bf Step 9.} Evaluate $w_{i_1,i_2}$ in $\Cc^n$ via the formula
\begin{align}
w_{i_1,i_2} = y_{\text{obs}} - H \cdot x_{i_1,i_2}.
\label{eq_w_i1}
\end{align}

{\bf Step 10.} Evaluate the real numbers $r_{i_1,i_2}$ and 
$\chi_{i_1,i_2}$ via the formulae
\begin{align}
r_{i_1,i_2} = \frac{ |w_{i_1,i_2}(1)|^2 + \dots + |w_{i_1,i_2}(n)|^2 }{\sigma^2}
\label{eq_r_i1}
\end{align}
and
\begin{align}
\chi_{i_1,i_2} = 1 - F_{\chi^2(2n)}(r_{i_1,i_2}),
\label{eq_chi_i1}
\end{align}
respectively,
where $F_{\chi^2(2n)}$ is the cdf of the $\chi^2_{2n}$-distribution
(see \eqref{eq_chi2_cdf_even}). \\

Thus, for every pair $(i_1,i_2)$ for which Steps 7-10 were performed,
the scheme produces $x_{i_1,i_2} \in X$
and the real numbers $r_{i_1,i_2}$ and $\chi_{i_1,i_2}$.
Suppose that, among these, the triplet
\begin{align}
\left( x_{\text{RaDe2}}, r_{\text{RaDe2}}, \chi_{\text{RaDe2}} \right)
=
\left( x_{i_1,i_2}, r_{i_1,i_2}, \chi_{i_1,i_2} \right)
\label{eq_x_RaDe2}
\end{align}
corresponds to the smallest $r_{i_1,i_2}$ 
(equivalently, to the largest $\chi_{i_1,i_2}$).

As a conclusion, we summarize the principal input and output parameters
of the decoding scheme described in this section.

{\bf Calling sequence:}
\begin{align}
\text{RaDe2\_search}\left( 
y_{\text{obs}}; \chi_{\text{stop}};
x_{\text{RaDe2}}, r, \chi \right).
\label{eq_RaDe2_calling}
\end{align}

{\bf Input parameters:}

-- $y_{\text{obs}}$ in $\Cc^n$: 

\quad the received message
(a noisy observation of $y(x_{\text{true}})$ defined via \eqref{eq_y_def}).

-- $\chi_{\text{stop}}$ (a real number between 0 and 1):

\quad the confidence threshold used in Step 6 (see \eqref{eq_chi_stop}). \\

{\bf Output parameters:}

-- $x_{\text{RaDe2}}$ in $X$: 

\quad the candidate for the transmitted message (see \eqref{eq_x_RaDe2}).

-- $r$ (a positive real number):

\quad the squared norm of the noise 
     (provided that $x_{\text{true}} = x_{\text{RaDe2}}$), see \eqref{eq_r_i1}.

-- $\chi$ (a real number between 0 and 1):

\quad the confidence in $x_{\text{RaDe2}}$ (see \eqref{eq_chi_i1}).

\subsubsection{RaDe2: full algorithm}
\label{sec_RaDe2_full}

In this section, we describe a decoding algorithm
whose basic step is the 
scheme described above (see Sections~\ref{sec_RaDe2_informal},
\ref{sec_RaDe2_detailed}).
This algorithm is closely related to RaDe1
(see Section~\ref{sec_RaDe1_full}), except that the basic
step of the latter (see Section~\ref{sec_RaDe1_detailed})
is replaced with the one from Section~\ref{sec_RaDe2_detailed}.

{\bf Additional parameters.}

-- \text{min\_RaDe2}: a positive integer
(minimal number of iterations)

-- $\text{max\_RaDe2} \geq \text{min\_RaDe2}$: a positive integer
(maximal number of iteration)

-- $\chi_{\text{thresh}}$: a real number between 0 and 1
(a confidence threshold)

-- $\chi_{\text{stop}}$: 
a real number between 0 and 1 
(the early exit confidence; see \eqref{eq_chi_stop}).

{\bf Description.}

{\bf 1.} For each $1 \leq j < \text{min\_RaDe2}$, 
call $\text{RaDe2\_search}\left( y_{\text{obs}}; 
\chi_{\text{stop}}; x_j, r_j, \chi_j\right)$.

{\bf 2.} Among all $(x_j,r_j,\chi_j)$, select the triplet
\begin{align}
\left(x_{\text{RaDe2}}, r_{\text{RaDe2}}, \chi_{\text{RaDe2}}\right) =
(x_i,r_i,\chi_i)
\label{eq_trip_RaDe2}
\end{align}
that corresponds to the smallest $r_i$.

{\bf 3.} For each $\text{min\_RaDe2} \leq j \leq \text{max\_RaDe2}$:

\quad {\bf 3a.} 
call $\text{RaDe2\_search}\left( y_{\text{obs}}; 
\chi_{\text{stop}}; x_j, r_j, \chi_j\right)$.

\quad {\bf 3b.}
if $r_j < r_{\text{RaDe2}}$, set
\begin{align}
\left(x_{\text{RaDe2}}, r_{\text{RaDe2}}, \chi_{\text{RaDe2}}\right) =
(x_j,r_j,\chi_j).
\label{eq_trip_RaDe22}
\end{align}

\quad {\bf 3c.} if $\chi_{\text{RaDe2}} > \chi_{\text{thresh}}$,
stop. \\

{\bf Comment.} The algorithm conducts at least min\_RaDe2 iterations
of the basic scheme. Then, if the confidence in the best guess
is high enough (compared to $\chi_{\text{thresh}}$), 
the algorithm stops. Otherwise, the basic scheme
is called iteratively until the confidence is high enough
{\it or} the total number of iterations reaches max\_RaDe2.

To conclude, we summarize the principal input and output parameters
of the algorithm described in this section.

{\bf Calling sequence:}
\begin{align}
\text{RaDe2\_all}\left( y_{\text{obs}}; 
\text{min\_RaDe2}, \text{max\_RaDe2}, \chi_{\text{thresh}}, 
\chi_{\text{stop}};
x_{\text{RaDe2}}, r_{\text{RaDe2}}, \chi_{\text{RaDe2}}\right).
\label{eq_RaDe2_all_calling}
\end{align}

{\bf Input parameters:}

-- $y_{\text{obs}}$ in $\Cc^n$ (see Section~\ref{sec_RaDe2_detailed}). 

-- min\_RaDe2: see above.

-- max\_RaDe2: see above.

-- $0 < \chi_{\text{thresh}} < 1$: see above. 

-- $0 < \chi_{\text{stop}} < 1$:  the confidence threshold
used in \eqref{eq_chi_stop}. \\

{\bf Output parameters:}

-- $x_{\text{RaDe2}}$ in $X$ (see \eqref{eq_trip_RaDe2}).

-- $r_{\text{RaDe2}}$ (a positive real number): see \eqref{eq_trip_RaDe2}.

-- $\chi_{\text{RaDe2}}$ (a real number between 0 and 1): 
see \eqref{eq_trip_RaDe2}.

\subsubsection{RaDe2: cost and memory requirements}
\label{sec_RaDe2_cost}

In this section, we describe the cost and memory requirements
of the basic scheme described in Sections~\ref{sec_RaDe2_informal}, 
\ref{sec_RaDe2_detailed} above.

{\bf Memory requirements.} The memory requirements of the scheme
are or the order
\begin{align}
M_{\text{RaDe2}} = O\left( m^2 \cdot n + n^2 \right)
\label{eq_memory_RaDe2}
\end{align}
memory words - in other words, absolutely minimal
(compare for example to
\eqref{eq_memory_preprocess}; see also Section~\ref{sec_running}
for typical values of $m$ and $n$).

{\bf Cost.} The precomputation step (see Section~\ref{sec_RaDe1_detailed})
requires $O(n^3)$ operations, and should not be re-done until
$H$ changes.

On the other hand, each decoding step requires
\begin{align}
C_{\text{RaDe2}} = O\left( \alpha \cdot m^2 \cdot n^2 \right)
\label{eq_cost_RaDe2}
\end{align}
operations, where $0 < \alpha < 1$ is the proportion of 
pairs $(i_1,i_2)$ for which Steps 7-10 in Section~\ref{sec_RaDe2_detailed}
are performed
(e.g. for which $\chi$ defined via \eqref{eq_rchi}
 is greater than $\chi_{\text{stop}}$).
 Therefore, the full algorithm described in 
Section~\ref{sec_RaDe2_full} requires between
\begin{align}
C_{\text{RaDe2,best}} = O\left( \text{min\_RaDe2} \cdot m^2 \cdot n^2 \right)
\label{eq_cost_RaDe2_best}
\end{align}
and
\begin{align}
C_{\text{RaDe2,worst}} = O\left( \text{max\_RaDe2} \cdot m^2 \cdot n^2 \right)
\label{eq_cost_RaDe2_worst}
\end{align}
operations.

{\bf Observation.} By comparing \eqref{eq_cost_RaDe2} to \eqref{eq_cost_RaDe1}
in Section~\ref{sec_RaDe1_cost}, we observe that the basic step of RaDe2
(see Section~\ref{sec_RaDe2_detailed}) is typically
somewhat slower than the basic
step of RaDe1
(see Section~\ref{sec_RaDe1_detailed})
roughly by a factor of $\alpha \cdot m$, 
where $0<\alpha<1$ is a real number
(see \eqref{eq_cost_RaDe2} above). Obviously, $\alpha$ depends,
among other things, on the parameter $\chi_{\text{stop}}$
(see Section~\ref{sec_RaDe2_detailed}).

{\bf Success rate.} Obviously, the success rate of RaDe2
depends on various parameters of the problems (e.g. $H, \sigma, X$)
as well as on the parameters of the algorithm
(see Section~\ref{sec_RaDe2_full}). In addition, since the scheme
is randomized, for each set of conditions, there is a certain probability
of failure (that decreases with the number
of iterations of the basic step;
see also Remark~\ref{rem_RaDe2_stuck} in Section~\ref{sec_exp_4}).
Obviously, there is a trade-off between the probability of success
and number of iterations.
In addition, compared to RaDe1 from Section~\ref{sec_RaDe1_full},
RaDe2 is, in general, slower, but has a higher success rate
(see Section~\ref{sec_numerical} below).

Also, in Section~\ref{sec_numerical}  we demonstrate
the performance of the scheme via several experiments.

\subsection{Supercharging}
\label{sec_super}

In this section, we describe a procedure that should not be used
on its own, but rather as an additional step after
the algorithms from previous sections, with the goal
to improve their output
at a moderate
computational cost.

\subsubsection{Supercharging: informal description}
\label{sec_super_informal}

Either of the algorithms RaDe1, RaDe2 
(see Sections~\ref{sec_RaDe1}, \ref{sec_RaDe2}) 
computes a "candidate" $x$ 
for $x_{\text{best}}$ defined via \eqref{eq_x_best};
in addition, it evaluates the confidence $\chi$ in this candidate
(see \eqref{eq_RaDe1_calling}, \eqref{eq_RaDe2_calling}).
This candidate can be improved (e.g. replaced by a better candidate)
by a procedure that we call "supercharging".

The idea behind supercharging is based on the observation
that even if $x$ is not equal to $x_{\text{best}}$, the latter
can still be among several nearest neighbors of $x$ in $X$.
In other words, we fix the integer $k_1$ of nearest neighbors
that we wish to inspect and find $k_1$ nearest neighbors
$x_1,\dots,x_{k_1}$ of $x$ in $X$
(see Section~\ref{sec_nn_x}).
Among these, we find the one that minimizes the distance between
$H \cdot x_i$ and $y_{\text{obs}}$.

\subsubsection{Supercharging: detailed description}
\label{sec_supercharging_detailed}

Suppose that $y_{\text{obs}}$ in $\Cc^n$
is the observed value of $x_{\text{true}}$
(see Section~\ref{sec_task}), that 
$x$ in $X$ is a candidate for $x_{\text{best}}$
evaluated by either 
RaDe1 or RaDe2
(see Sections~\ref{sec_RaDe1}, \ref{sec_RaDe2}), 
that $r$ is the (normalized) squared norm of the noise,
and that $\chi$ is the confidence in $x$ 
(see \eqref{eq_trip}, \eqref{eq_trip_RaDe2}). Suppose also
that $k_1>0$ is the number of nearest neighbors of $x$
that we want to inspect. Supercharging consists of the following steps.

{\bf Step 1.} Evaluate the $k_1$ nearest neighbors $x_1,\dots,x_{k_1}$
of $x$ in $X$ (see Section~\ref{sec_nn_x}).

{\bf Step 2.} For every $i=1,\dots,k_1$, evaluate the real number $r_i$
via the formula
\begin{align}
r_i = \frac{ \| H \cdot x_i - y_{\text{obs}} \|^2 }{ \sigma^2 }.
\label{eq_ri_super}
\end{align}

{\bf Step 3.} Find the minimum $r_i$ among $r_1,\dots,r_{k_1}$.

{\bf Step 4.} If $r_i < r$, evaluate $x_{\text{super}}$ in $X$ and
the real numbers $r_{\text{super}}, \chi_{\text{super}}$ via the formulae
\begin{align}
\label{eq_x_super}
& x_{\text{super}} = x_i, \\
\label{eq_r_super}
& r_{\text{super}} = r_i, \\
\label{eq_chi_super}
& \chi_{\text{super}} = 1 - F_{\chi^2(2n)}(r_i),
\end{align}
where $F_{\chi^2(2n)}$ is the cdf of the $\chi^2(2n)$ distribution
defined via \eqref{eq_chi2_cdf_even}.

As a conclusion, we summarize the principal input and output parameters
of supercharging.

{\bf Calling sequence:}
\begin{align}
\text{super}\left( y_{\text{obs}}; k_1, x, r, \chi; 
x_{\text{super}}, r_{\text{super}}, \chi_{\text{super}} \right).
\label{eq_super_calling}
\end{align}

{\bf Input parameters:}

-- $y_{\text{obs}}$ in $\Cc^n$: 

\quad the received message
(a noisy observation of $y(x_{\text{true}})$ defined via \eqref{eq_y_def}).

-- $k_1$ (a positive integer):

\quad the number of nearest neighbors to use.

-- $x$ in $X$:

\quad the candidate for $x_{\text{best}}$ 
(see \eqref{eq_RaDe1_calling}, \eqref{eq_RaDe2_calling}).

-- $r$ (a positive real number):

\quad the normalized squared norm of the noise 
(see \eqref{eq_RaDe1_calling}, \eqref{eq_RaDe2_calling}).

-- $\chi$ (a real number between 0 and 1):

\quad the confidence in $x$
(see \eqref{eq_RaDe1_calling}, \eqref{eq_RaDe2_calling}). \\

{\bf Output parameters:}

-- $x_{\text{super}}$ in $X$: 

\quad the improved candidate (see \eqref{eq_x_super}).

-- $r_{\text{super}}$ (a positive real number):

\quad the corresponding normalized squared norm of the noise 
      (see \eqref{eq_r_super}).

-- $\chi_{\text{super}}$ (a real number between 0 and 1):

\quad the confidence in $x_{\text{super}}$ (see \eqref{eq_chi_super}).

\subsubsection{Supercharging: cost and memory requirements}
\label{sec_super_cost}
The memory requirements and number of operations of supercharging
are given, respectively, 
by the formulae \eqref{eq_memory_nn}, \eqref{eq_cost_nn}
in Section~\ref{sec_nn_x}
(obviously, with $k$ being replaced by $k_1$).
 
However, often
this cost can be reduced roughly by the factor of $n$ as follows.
Suppose, for the sake of concreteness, that the constellation $C$
is given via \eqref{eq_running_c} in Section~\ref{sec_running}.
Then, every $x$ in $X$ has $2 \cdot n$ nearest neighbors 
$x_1,\dots,x_{2n}$ in $X$
all of which are at exactly the same distance from $x$.
Each of $x_1,\dots,x_{2n}$ differs from $x$ at precisely one coordinate;
in other words, for every $i=1,\dots,2n$, the difference
$x-x_i$ has only one non-zero coordinate. In particular, 
each of $r_i$ defined via \eqref{eq_ri_super} can be evaluated
in $O(n)$ rather than $O(n^2)$ operations. In other words,
if, for example,
\begin{align}
k_1 = 2 \cdot n,
\label{eq_k1_special}
\end{align}
then the total cost of supercharging is
\begin{align}
C_{\text{super}}(2 \cdot n) = O(n^2)
\label{eq_cost_super}
\end{align}
(rather than $O(n^3)$) operations.



\section{Numerical Results}
\label{sec_numerical}
In this section,
we illustrate the 
performance of the schemes from Section~\ref{sec_num_algo}
via several numerical experiments. All the calculations were
implemented in FORTRAN (the Lahey 95 LINUX version), and were
carried out in
double precision,
on a standard laptop computer
with DualCore CPU 2.53 GHz and 2.9GB RAM.

\subsection{Experiments: basic structure}
\label{sec_num_structure}

All the experiments are build around the model
described in Section~\ref{sec_task}. In addition,
we choose our parameters based on the running example
from Section~\ref{sec_running}.

Each experiment consists of five stages described below.

\subsubsection{Stage 1: preparation}
\label{sec_stage_1}

On this stage, we perform operations that do not
depend on the channel matrix $H$, let alone
the transmitted (or received) message.

We start by selecting the number
of transmitters and receivers $n$ 
(see \eqref{eq_running_n}), 
the constellation size $m$ 
(see \eqref{eq_running_m})
and the constellation $C$
(see \eqref{eq_running_c}).
In particular, the constellation $C$ 
is the same in all experiments, as is the collection $X$ of all
possible messages (see \eqref{eq_bigx_def}).

Next, we select the integer $k>0$ and find the list of $k$ nearest neighbors
of $x_0$ (see \eqref{eq_x_special}) in $X$ by brute force.
For any $x$ in $X$, this list allows us to compute 
its $k$ nearest neighbors in $O(k \cdot n)$ operations
(see Section~\ref{sec_nn_x}, in particular \eqref{eq_x_special_j}).

\subsubsection{Stage 2: $H$ and related quantities}
\label{sec_stage_2}

On this stage, we form the $n$ by $n$ complex matrix $H$
and carry out some computations that depend only on $H$.

First, we generate $H$
according to \eqref{eq_running_h} (i.e. by drawing $n^2$ independent
samples from $N\Cc(0,1)$ and using them as entries of $H$).

Then, we compute the singular value decomposition 
of $H$ (e.g. the matrices $U,\Sigma,V$ in the notation of
Theorem~\ref{thm_svd}). We use $\Sigma$ and $V$ to evaluate
the real numbers $S_{n-2}(1), \dots, S_{n-2}(n)$,
$S_{n-1}(1),\dots,S_{n-1}(n)$ (see \eqref{eq_skj_2})
and $s_{n-1}(1),\dots,s_{n-1}(n)$, $s_n(1),\dots,s_n(n)$
(see \eqref{eq_skj}). Also, we find the integer $1 \leq j \leq n$
that corresponds to the minimal $S_{n-1}(j)$ among
$S_{n-1}(1),\dots,S_{n-1}(n)$
(see the precomputation step in Section~\ref{sec_RaDe1_detailed}).
In addition, we find the integers $1 \leq j_1, j_2 \leq n$
such that $S_{n-2}(j_1)$ and $S_{n-2}(j_2)$ are the two smallest
values among $S_{n-2}(1), \dots, S_{n-2}(n)$
(see the precomputation step in Section~\ref{sec_RaDe2_detailed}).

Finally, we select the real number
$\sigma>0$ (roughly of the same order of magnitude
that the smallest singular value $\sigma_n$ of $H$:
see
Theorem~\ref{thm_svd} and Remark~\ref{rem_best}).

In Table~\ref{t:stage_12} 
we summarize the input and output parameters of all computations
described in Sections~\ref{sec_stage_1}, \ref{sec_stage_2}.

\begin{table}[htbp]
\begin{center}
\begin{tabular}{c|l}
Variable & Details \\
\hline
$n$            &  dimensionality of received/transmitted message 
                  \eqref{eq_running_n} \\
$m$            &  size of constellation \eqref{eq_running_m} \\
$C$            &  the complex constellation \eqref{eq_running_c} \\
$X$            &  the collection of all possible messages \eqref{eq_bigx_def} \\
$k$            &  the number of nearest neighbors \\
$\hat{x}_1,\dots,\hat{x}_k$ &
 the nearest neighbors of $x_0$ \eqref{eq_x_special} in $X$ \\
\hline
$H$            & the $n$ by $n$ complex channel matrix \eqref{eq_running_h} \\
$U,\Sigma,V$   & the SVD of $H$ (see Theorem~\ref{thm_svd}) \\
$v_1,\dots,v_n$& the columns of $V$ (vectors in $\Cc^n$) \\
$\sigma_1,\dots,\sigma_n$ & the singular values of $H$ 
(positive real numbers) \\
$s_{n-1}(1),\dots,s_{n-1}(n)$ & noise deviations, see \eqref{eq_skj} \\
$s_n(1),\dots,s_n(n)$         & noise deviations, see \eqref{eq_skj} \\
$S_{n-2}(1), \dots, S_{n-2}(n)$ & noise deviations, see \eqref{eq_skj_2} \\
$S_{n-1}(1),\dots,S_{n-1}(n)$ & noise deviations, see \eqref{eq_skj_2} \\
$j$ & the index of the minimal $S_{n-1}(j)$ \\
$j_1,j_2$ & the indices of the two smallest $S_{n-2}(j_1), S_{n-2}(j_2)$ \\
$\sigma$ & the deviation of noise in \eqref{eq_y_def} \\
\end{tabular}
\end{center}
\caption{\it Basic experiment: list of input variables}
\label{t:stage_12}
\end{table}

\subsubsection{Stage 3: generation of messages}
\label{sec_stage_3}
During this stage, we choose, more or less arbitrarily,
the integer $L>0$, and generate $L$ points
$x_{\text{true},1}, \dots, x_{\text{true},L}$ in $X$
by taking $L$ independent samples from the uniform distribution 
on $X$. (In the language of the model from
Section~\ref{sec_task}, each $x_{\text{true},i}$ represents
a message to be transmitted).

For each $i=1,\dots,L$, we generate $y_{\text{obs},i}$ from
$x_{\text{true},i}$ according to \eqref{eq_y_def} as follows:
we sample a complex normal standard random vector $\hat{z}_i$ 
in $\Cc^n$ (i.e. from the distribution $N\Cc(0_n,I_n)$)
and evaluate $y_{\text{obs},i}$ via the formula
\begin{align}
y_{\text{obs},i} = H \cdot x_{\text{true},i} + \sigma \cdot \hat{z}_i,
\label{eq_y_obs_i}
\end{align}
where $\sigma$ is that from Table~\ref{t:stage_12}. Needless to say,
all $\hat{z}_i$'s are sampled independently of each other.
In the language of the model from Section~\ref{sec_task},
each $y_{\text{obs},i}$ represents a received message corrupted by 
Gaussian noise of coordinate-wise standard deviation $\sigma$.

Our ultimate goal is to decode each received message $y_{\text{true},i}$
by computing the maximal likelihood estimate $x_{\text{best},i}$
to $x_{\text{true},i}$, in the sense of \eqref{eq_x_best}. We do so
by means of several numerical schemes from Section~\ref{sec_num_algo}.
Some of these schemes have a finite probability of failure
(e.g. they compute a candidate for $x_{\text{best},i}$ which
is equal to the latter with probability less than one).
Obviously, none of the schemes cannot use either $x_{\text{true},i}$
or $\hat{z}_i$ (e.g. both the original message and the noise
are assumed to be unknown). In other words, each scheme
is allowed to access only $y_{\text{obs},i}$ as well
as the variables from Table~\ref{t:stage_12}.

\subsubsection{Stage 4: decoding}
\label{sec_stage_4}

On this stage, we decode each of the messages 
$y_{\text{obs},1}, \dots, y_{\text{obs},L}$ by means of several schemes
from Section~\ref{sec_num_algo}.
More specifically, we use the following schemes:

{\bf 1.} Brute force (see Section~\ref{sec_brute_force}).

{\bf 2.} Nearest neighbors search in $X$ (see Section~\ref{sec_nn_x}).

{\bf 3.} Randomized Decoder 1 (see Section~\ref{sec_RaDe1}).

{\bf 4.} Randomized Decoder 2 (see Section~\ref{sec_RaDe2}).

Optionally, we improve the results RaDe1 (or RaDe2)
by supercharging (see Section~\ref{sec_super}).

In addition to $y_{\text{obs}_i}$ and the variables from 
Table~\ref{t:stage_12}, RaDe1 and RaDe2 
also receive several "tuning" parameters listed
in Table~\ref{t:some_params} below 
(see Sections~\ref{sec_RaDe1_full}, \ref{sec_RaDe2_full}
for details).

\begin{table}[htbp]
\begin{center}
\begin{tabular}{c|l}
Parameter & Details \\
\hline
$\chi_{\text{thresh}}$ & the threshold confidence 
(see Section~\ref{sec_RaDe1_detailed}, 
\ref{sec_RaDe2_detailed}) \\
min\_RaDe1 & minimal number of iterations of RaDe1 
(see Section~\ref{sec_RaDe1_full}) \\
max\_RaDe1 & maximal number of iterations of RaDe1 
(see Section~\ref{sec_RaDe1_full}) \\
\hline
$\chi_{\text{stop}}$ & 
the early exit confidence (see \eqref{eq_chi_stop}
in Section~\ref{sec_RaDe2_detailed}) \\
min\_RaDe2 & minimal number of iterations of RaDe2
(see Section~\ref{sec_RaDe2_full}) \\
max\_RaDe2 & maximal number of iterations of RaDe2
(see Section~\ref{sec_RaDe2_full}) \\
\hline
isuper & whether to perform supercharging (1) or not (0) \\
$k_1$ & the number of nearest neighbors in supercharging 
(see \eqref{eq_k1_special}).
\end{tabular}
\end{center}
\caption{\it Basic experiment: list of "tuning" parameters}
\label{t:some_params}
\end{table}
For every $i=1,\dots,L$, we proceed as follows: 

\vspace{0.05in}
{\bf 1.} Evaluate $x_{\text{best},i}$ from $y_{\text{obs},i}$ by
brute force (see Section~\ref{sec_brute_force}).

{\bf Comment.} This is the "correct answer" (see, however,
Remark~\ref{rem_best} in Section~\ref{sec_brute_force}).

\vspace{0.05in}
{\bf 2.} Evaluate $\tilde{x}_{\text{obs},i}$ from
$y_{\text{obs},i}$ via \eqref{eq_tilde_x_obs}.

\vspace{0.05in}
{\bf 3.} Evaluate $x_{\text{nn}(k),i}$ from $y_{\text{obs},i}$
via \eqref{eq_x_nn} in
Section~\ref{sec_nn_x} by using $k$ nearest neighbors.

\vspace{0.05in}
{\bf 4.} Evaluate $x_{\text{RaDe1},i}$ from $y_{\text{obs},i}$
via RaDe1 search
(see Section~\ref{sec_RaDe1_full}).

\vspace{0.05in}
{\bf 5.} Evaluate $x_{\text{RaDe2},i}$ from $y_{\text{obs},i}$
via RaDe2 search
(see Section~\ref{sec_RaDe2_full}).

\vspace{0.05in} 
{\bf 6.} If $\text{isuper}=1$, evaluate $x_{\text{super},i}$ 
from the best
of $x_{\text{RaDe1},i}, x_{\text{RaDe2},i}$
by supercharging
(see \eqref{eq_best_two}, \eqref{eq_decode_super} below,
and also
Section~\ref{sec_super}).

In other words, for every $i=1,\dots,L$, we execute
\begin{align}
& \text{brute}(y_{\text{obs},i}; x_{\text{best},i}) \nonumber \\
& \text{nnx}
(y_{\text{obs},i}; k; \tilde{x}_{\text{obs},i}, x_{\text{nn}(k),i})
\nonumber \\
& \text{RaDe1\_all}\left( y_{\text{obs},i}; 
\text{min\_RaDe1}, \text{max\_RaDe1}, \chi_{\text{thresh}};
x_{\text{RaDe1},i}, r_{\text{RaDe1},i}, \chi_{\text{RaDe1},i}\right)
\nonumber \\
& \text{RaDe2\_all}\left( y_{\text{obs},i}; 
\text{min\_RaDe2}, \text{max\_RaDe2}, \chi_{\text{thresh}},
\chi_{\text{stop}};
x_{\text{RaDe2},i}, r_{\text{RaDe2},i}, \chi_{\text{RaDe2},i}\right)
\label{eq_decode_all}
\end{align}
(see \eqref{eq_brute_calling},
\eqref{eq_nnx_calling},
\eqref{eq_RaDe1_all_calling},
\eqref{eq_RaDe2_all_calling}).
In addition, if $\text{isuper} = 1$, we define
$x_i$ in $X$ and the real numbers $r_i, \chi_i$ via the formula
\begin{align}
(x_i,r_i,\chi_i) = \begin{cases}
\left(
x_{\text{RaDe1},i}, r_{\text{RaDe1},i}, \chi_{\text{RaDe1},i}
\right) & 
\text{ if } r_{\text{RaDe1},i} < r_{\text{RaDe2},i}, \\
\left(
x_{\text{RaDe2},i}, r_{\text{RaDe2},i}, \chi_{\text{RaDe2},i}
\right) & 
\text{ otherwise}
\end{cases}
\label{eq_best_two}
\end{align}
and execute
\begin{align}
\text{super}\left( y_{\text{obs},i}; k_1, x_i, r_i, \chi_i; 
x_{\text{super},i}, r_{\text{super},i}, \chi_{\text{super},i} \right)
\label{eq_decode_super}
\end{align}
(see \eqref{eq_super_calling}).

\subsubsection{Stage 5: evaluation of statistics}
\label{sec_stage_5}

In Section~\ref{sec_stage_4} above, we describe
the decoding of a single received message
by means of several schemes.
In this section,
we describe a way to
compare the performance of these schemes via the evaluation
of various statistics
(essentially, via Monte Carlo simulations). 
Briefly speaking, we evaluate how much
time, on average, it took for each algorithm to compute its output,
and in what proportion of cases the output was correct.

Each statistic is evaluated for every $i=1,\dots,L$. In addition,
we also average over several matrices $H$
(corresponding to the same set of parameters).

\subsection{Experiment 1}
\label{sec_exp_1}

{\bf Description.}
In this experiment, we proceed as follows. For each $n=6,7,8$, 
we generate five $n$ by $n$ complex matrices $H$ (see
Section~\ref{sec_stage_2} above). For each such matrix
and each $\sigma = 0.25, 0.5, 0.75, 1, 1.25$, 
we generate $L=1000$ messages $x_{\text{true}}$ in $X$, and for
each such $x_{\text{true}}$ we generate $y_{\text{obs}}$ as
described in Section~\ref{sec_stage_3}. For each such $y_{\text{obs}}$,
we evaluate the maximal likelihood estimate $x_{\text{best}}$ of
$x_{\text{true}}$ via the brute force algorithm from
Section~\ref{sec_brute_force} (see also Section~\ref{sec_stage_4}).

Thus, for each $n=6,7,8$ and each $\sigma=0.25,0.5,0.75,1,1.25$,
we obtain 5,000 pairs $(x_{\text{true}}, v_{\text{best}})$.
We evaluate the proportion of cases when $x_{\text{true}}$
is equal to $x_{\text{best}}$. The results of this experiment
are displayed in Table~\ref{t:test21}.

\begin{table}[htbp]
\begin{center}
\begin{tabular}{c|c|c|c|c|c}
\multicolumn{1}{c|}{} & 
\multicolumn{5}{c}{$\sigma$} \\
\hline
$n$ & 0.25 & 0.5 & 0.75 & 1 & 1.25 \\
\hline
6 &  0.99840E+00 &  0.83740E+00 & 0.44200E+00 & 0.17900E+00 & 0.74800E-01 \\
7 &  0.99940E+00 &  0.87960E+00 & 0.49780E+00 & 0.20440E+00 & 0.77200E-01 \\
8 &  0.10000E+01 &  0.90900E+00 & 0.45600E+00 & 0.15220E+00 & 0.49000E-01 \\
\end{tabular}
\end{center}
\caption{\it Average proportion of $x_{\text{true}} = x_{\text{best}}$, five
channel matrices $H$ per each $n$. The number of messages is $L=1,000$
(per matrix). See Section~\ref{sec_exp_1}.}
\label{t:test21}
\end{table}

{\bf Observations.} Several observations can be made from
Table~\ref{t:test21}.

{\bf 1.} For each $n$, the average proportion of 
$x_{\text{best}} = x_{\text{in}}$ decreases monotonically with $\sigma$,
as expected (the larger $\sigma$ is the more likely $x_{\text{best}}$
is to be different from $x_{\text{true}}$). For example,
for $\sigma = 0.25$ this proportion is above $0.99$ for all $n=6,7,8$,
while for $\sigma = 1.25$ it might be as low as $0.05$ (for $n=8$).
See also Remark~\ref{rem_best} in Section~\ref{sec_task}.

{\bf 2.} In the view of the previous observation, for each
$n$ the values of $\sigma$ vary from that corresponding
to a relatively "easy" decoding task ($\sigma=0.25$)
to that corresponding to a 
"difficult but possible" decoding task ($\sigma=1.25$),
in the sense of Section~\ref{sec_task}.

Motivated by Observation 2, we will select the same values
of $\sigma$ in all numerical experiments below,
to investigate how the performance of various numerical schemes
for the solution of the decoding task depends on the 
standard deviation of noise and dimensionality of the matrix.

\subsection{Experiment 2}
\label{sec_exp_2}

In this experiment, we investigate the performance
of the 
nearest neighbors search in $X$
(see Section~\ref{sec_nn_x}).

{\bf Description.} For each $n=6,7,8$ and each $\sigma=0.25,0.5,0.75,1,1.25$,
we proceed as in Experiment 1 from Section~\ref{sec_exp_1}. Then,
for each received message $y_{\text{obs}}$ and each 
$k=1,2n+1,2n^2+1,n^3,n^4,n^5+1$,
we evaluate the estimate $x_{\text{nn}(k)}$ of $x_{\text{true}}$
via the algorithm from Section~\ref{sec_nn_x}
(see \eqref{eq_nnx_calling}). Then, among all cases when
$x_{\text{true}} = x_{\text{best}}$, we evaluate the proportion
of cases when the estimate $x_{\text{nn}(k)}$ is correct; in other words,
we compute the number $\text{prop}_{\text{nn}}(n,\sigma,k)$
defined via the formula
\begin{align}
\text{prop}_{\text{nn}}(n,\sigma,k) =
\frac{ \# \left\{ x_{\text{best}} = x_{\text{true}} = 
                  x_{\text{nn}(k)} \right\} }
     { \# \left\{ x_{\text{best}} = x_{\text{true}} \right\} }.
\label{eq_prop_nn_x}
\end{align}
As in Experiment 1 from Section~\ref{sec_exp_1}, to evaluate
each such proportion we use 5,000 messages 
(five matrices per $n$, $L=1,000$ messages per matrix).
\begin{remark}
According to \eqref{eq_prop_nn_x}, in this experiment we evaluate the
performance of the scheme only on those messages $x_{\text{true}}$
for which the maximal likelihood estimate $x_{\text{best}}$
coincides with $x_{\text{true}}$. 
\label{rem_exp_2}
\end{remark}

{\bf Tables.} The results of Experiment 2 are displayed
in Tables~\ref{t:test22_6}--\ref{t:test22_8}. These tables
correspond, respectively, to $n=6,7,8$, 
and contain $\text{prop}_{\text{nn}}(n,\sigma,k)$
defined via \eqref{eq_prop_nn_x} above
(in each table, rows correspond to the values of $\sigma$
and columns correspond to the values of $k$).

In addition, for each $n$ and $k$, we measure the CPU time
required by this scheme to evaluate $x_{\text{nn}(k)}$ for
$L=1,000$ messages (obviously, this CPU time essentially does not
depend on $\sigma$: see Section~\ref{sec_nn_x}). These CPU times
are listed in the last five columns of Table~\ref{t:test23_24}.
Each row of this table corresponds to $n=6,7,8$,
while the last five columns of this table correspond to
$k=2n, 2n^2, n^3, n^4, n^5$, respectively.

\begin{table}[htbp]
\begin{center}
\begin{tabular}{c|c|c|c|c|c|c}
\multicolumn{1}{c|}{} & 
\multicolumn{6}{c}{$k$} \\
\hline
$\sigma$ & 1 & $2n+1$ & $2n^2+1$ & $n^3$ & $n^4$ & $n^5+1$ \\
\hline
0.25 &  0.659E+00 &  0.825E+00 &  0.906E+00 &  0.939E+00 &  0.973E+00 &  0.991E+00 \\
0.5 &  0.242E+00 &  0.458E+00 &  0.645E+00 &  0.741E+00 &  0.852E+00 &  0.922E+00 \\
0.75 &  0.108E+00 &  0.281E+00 &  0.447E+00 &  0.559E+00 &  0.717E+00 &  0.862E+00 \\ 
1 &  0.693E-01 &  0.194E+00 &  0.362E+00 &  0.468E+00 &  0.660E+00 &  0.812E+00 \\
1.25 &  0.374E-01 &  0.144E+00 &  0.286E+00 &  0.385E+00 &  0.596E+00 &  0.805E+00 \\
\end{tabular}
\end{center}
\caption{\it Success rate of nearest neighbors search in $X$
(see \eqref{eq_prop_nn_x} in Section~\ref{sec_exp_2}).
Corresponds to Experiment 2 with $n=6$.
}
\label{t:test22_6}
\end{table}

\begin{table}[htbp]
\begin{center}
\begin{tabular}{c|c|c|c|c|c|c}
\multicolumn{1}{c|}{} & 
\multicolumn{6}{c}{$k$} \\
\hline
$\sigma$ & 1 & $2n+1$ & $2n^2+1$ & $n^3$ & $n^4$ & $n^5+1$ \\
\hline
0.25 &  0.603E+00 &  0.794E+00 &  0.897E+00 &  0.926E+00 &  0.971E+00 &  0.993E+00 \\
0.5 &  0.226E+00 &  0.418E+00 &  0.568E+00 &  0.658E+00 &  0.808E+00 &  0.917E+00 \\
0.75 &  0.996E-01 &  0.236E+00 &  0.374E+00 &  0.483E+00 &  0.648E+00 &  0.796E+00 \\
1 &  0.489E-01 &  0.139E+00 &  0.262E+00 &  0.359E+00 &  0.524E+00 &  0.679E+00 \\
1.25 &  0.311E-01 &  0.907E-01 &  0.184E+00 &  0.251E+00 &  0.425E+00 &  0.598E+00 \\
\end{tabular}
\end{center}
\caption{\it Success rate of nearest neighbors search in $X$
(see \eqref{eq_prop_nn_x} in Section~\ref{sec_exp_2}).
Corresponds to Experiment 2 with $n=7$.
}
\label{t:test22_7}
\end{table}

\begin{table}[htbp]
\begin{center}
\begin{tabular}{c|c|c|c|c|c|c}
\multicolumn{1}{c|}{} & 
\multicolumn{6}{c}{$k$} \\
\hline
$\sigma$ & 1 & $2n+1$ & $2n^2+1$ & $n^3$ & $n^4$ & $n^5+1$ \\
\hline
0.25 &  0.430E+00 &  0.569E+00 &  0.655E+00 &  0.699E+00 &  0.765E+00 &  0.821E+00 \\
0.5 &  0.112E+00 &  0.244E+00 &  0.369E+00 &  0.450E+00 &  0.566E+00 &  0.666E+00 \\
0.75 &  0.298E-01 &  0.101E+00 &  0.197E+00 &  0.287E+00 &  0.442E+00 &  0.575E+00 \\
1 &  0.920E-02 &  0.420E-01 &  0.125E+00 &  0.192E+00 &  0.311E+00 &  0.466E+00 \\
1.25 &  0.816E-02 &  0.245E-01 &  0.571E-01 &  0.106E+00 &  0.196E+00 &  0.363E+00 \\
\end{tabular}
\end{center}
\caption{\it Success rate of nearest neighbors search in $X$
(see \eqref{eq_prop_nn_x} in Section~\ref{sec_exp_2}).
Corresponds to Experiment 2 with $n=8$.
}
\label{t:test22_8}
\end{table}

\begin{table}[htbp]
\begin{center}
\begin{tabular}{c|c|c|c|c|c|c|c}
\multicolumn{1}{c|}{} & 
\multicolumn{1}{c|}{} & 
\multicolumn{1}{c|}{} & 
\multicolumn{4}{c}{$k$} \\
\hline
$n$ & RaDe1 & RaDe2    & $2n$ & $2n^2$ & $n^3$ & $n^4$ & $n^5$ \\
\hline
6 & 0.46E-1 & 0.14E+0 & 0.20E-1 & 0.75E-1 & 0.22E+0 & 0.13E+1 & 0.76E+1 \\
7 & 0.50E-1 & 0.12E+0 & 0.26E-1 & 0.14E+0 & 0.40E+0 & 0.27E+1 & 0.19E+2 \\
8 & 0.54E-1 & 0.15E+0 & 0.33E-1 & 0.27E+0 & 0.80E+0 & 0.56E+1 & 0.45E+2 \\
\end{tabular}
\end{center}
\caption{\it CPU time of various decoding schemes 
(in seconds). The number of messages is $L=1000$.
For RaDe1 and RaDe2, we used $\sigma=0.75$.}
\label{t:test23_24}
\end{table}

{\bf Observations.} We make the following observations from
Tables~\ref{t:test22_6}--\ref{t:test22_8}.

{\bf 1.} As expected, for each $n$ and $\sigma$ the performance
of the scheme improves as $k$ increases (obviously, this
improvement comes at the cost of additional CPU time). For example,
for $n=8$ and $\sigma=0.75$  the proportion of correct 
"guesses" grows from about 2\% to about 58\%
as $k$ increases from 1 to $n^5$.

{\bf 2.} As expected, for each $n$ and $k$ the performance of the
scheme deteriorates as $\sigma$ increases (noise of a larger
standard deviation makes
the decoding task more difficult).

{\bf 3.} Typically, for each $\sigma$ and comparable values of $k$,
the performance of the scheme deteriorates as $n$ increases.

{\bf 4.} When $n=8$ and $\sigma = 1, 1.25$, the scheme returns
the correct answer in less than half of all cases for all values
of $k$ in Table~\ref{t:test22_8}. In other words, when, for 
example, $\sigma=1$, in more than half of all cases
$x_{\text{true}}$ is not among as many as $32,000$
nearest neighbors of $H^{-1} \cdot y_{\text{obs}}$ in $X$,
even when $x_{\text{true}} = x_{\text{best}}$
(see \eqref{eq_tilde_x_obs} in Section~\ref{sec_nn_x}).

Some additional observations can be made from Table~\ref{t:test23_24}.

{\bf 5.} Naturally, the CPU time of nearest neighbor search in $X$
scales roughly as one would expect from
 \eqref{eq_cost_nn} in Section~\ref{sec_nn_x}.

{\bf 6.} Typically, even a noticeable increase in the  CPU time
(e.g. taking a larger $k$) results 
in a fairly modest improvement in performance.
For example, when $n=8$ and $\sigma=1$, using $k=8^4$ nearest neighbors
to decode the message results in about 31\% of correct guesses
(and it takes about 6 seconds per 1000 messages). Increasing
the number of nearest neighbors by a factor of eight allows
one to determine about half as many additional messages correctly
(e.g. about 46\% overall),
while the CPU time goes up to about 45 seconds.

\subsection{Experiment 3}
\label{sec_exp_3}

In this experiment, we investigate the performance
of Randomized Decoder 1
(see Section~\ref{sec_RaDe1}).

{\bf Description.} For each $n=6,7,8$ and each $\sigma=0.25,0.5,0.75,1,1.25$,
we proceed as in Experiment 1 (see Section~\ref{sec_exp_1}). 
Then, for each received message $y_{\text{obs}}$ and each $T=1,\dots,7$,
we evaluate the estimate $x_{\text{RaDe1}}(T)$ of $x_{\text{true}}$
from $y_{\text{obs}}$
via RaDe1 (see Section~\ref{sec_RaDe1}), with the following parameters
(see \eqref{eq_RaDe1_calling} in Section~\ref{sec_RaDe1_full},
and also Table~\ref{t:some_params} in Section~\ref{sec_stage_3}):
\begin{align}
& \text{min\_RaDe1} = T, \nonumber \\
& \text{max\_RaDe1} = T, \nonumber \\
& \chi_{\text{thresh}}: \text{not applicable}
\label{eq_exp_3_params}
\end{align}
(the parameter $\chi_{\text{thresh}}$ is not applicable since
 $\text{min\_RaDe1} = \text{max\_RaDe1} = T$, e.g. the number
of basic iterations is always exactly $T$). Then, 
we define the integer $k_1$ via the formula
\begin{align}
k_1 = 2 \cdot n^2.
\label{eq_exp_3_k1}
\end{align}
and evaluate
the estimate $x_{\text{super}}(T)$ of $x_{\text{true}}$
from $y_{\text{obs}}$ via RaDe1
(where the parameters are defined via \eqref{eq_exp_3_params})
followed by supercharging (see Section~\ref{sec_super},
in particular \eqref{eq_super_calling}); during the supercharging
step, $k_1$ nearest neighbors are used. In other words,
$x_{\text{RaDe1}}(T)$ and $x_{\text{super}}(T)$ are obtained 
from $y_{\text{obs}}$ via
calling
\begin{align}
& 
\text{RaDe1\_search}\left( y_{\text{obs}}; x_{\text{RaDe1}}(T), r, \chi\right)
\nonumber \\
& 
\text{RaDe1\_search}\left( y_{\text{obs}}; x, r, \chi\right)
\nonumber \\
&
\text{super}\left( y_{\text{obs}}; k_1, x, r, \chi; 
x_{\text{super}}(T), r_{\text{super}}, \chi_{\text{super}} \right)
\label{eq_exp_3_calling}
\end{align}
(see \eqref{eq_RaDe1_calling}, \eqref{eq_super_calling}),
where the parameters are defined via \eqref{eq_exp_3_params}
and \eqref{eq_exp_3_k1}.

Thus, for each $n$ and $\sigma$, we have 5,000 messages $y_{\text{obs}}$
(five matrices per $n$,
$L=1,000$ messages per matrix,
as in Experiments 1,2 from Sections~\ref{sec_exp_1}, \ref{sec_exp_2});
for each such message, we obtain 14 estimates
$x_{\text{RaDe1}}(1), \dots, x_{\text{RaDe1}}(7)$ and
$x_{\text{super}}(1), \dots, x_{\text{super}}(7)$
of $x_{\text{true}}$. For each such estimate, we evaluate
the proportion of the cases in which this estimate is correct
(provided that $x_{\text{true}} = x_{\text{best}}$, see
Section~\ref{sec_exp_1}). In other words, we compute
$\text{prop}_{\text{RaDe1}}(n,\sigma,T)$ and
$\text{prop}_{\text{super}}(n,\sigma,T)$ via the formulae
\begin{align}
\text{prop}_{\text{RaDe1}}(n,\sigma,T) =
\frac{ \# \left\{ x_{\text{best}} = x_{\text{true}} = 
                  x_{\text{RaDe1}(T)} \right\} }
     { \# \left\{ x_{\text{best}} = x_{\text{true}} \right\} }
\label{eq_exp_3_prop_RaDe1}
\end{align}
and
\begin{align}
\text{prop}_{\text{super}}(n,\sigma,T) =
\frac{ \# \left\{ x_{\text{best}} = x_{\text{true}} = 
                  x_{\text{super}(T)} \right\} }
     { \# \left\{ x_{\text{best}} = x_{\text{true}} \right\} },
\label{eq_exp_3_prop_super}
\end{align}
respectively (see also Remark~\ref{rem_exp_2} in Section~\ref{sec_exp_2}).

{\bf Tables.} The results of Experiment 3 are displayed
in Table~\ref{t:test25}. 
In this table, the rows correspond to all pairs of $n=6,7,8$
and $T=1,\dots,7$, and the columns correspond to 
$\sigma=0.25,0.5,0.75,1,1.25$. For each $n,\sigma,T$, the
table contains two entries that appear one under the other: 
$\text{prop}_{\text{RaDe1}}(n,\sigma,T)$
(above)
and $\text{prop}_{\text{super}}(n,\sigma,T)$ (below);
see \eqref{eq_exp_3_prop_RaDe1},
\eqref{eq_exp_3_prop_super}. For example, for $n=6$, $T=1$ and $\sigma=0.5$,
these proportions are equal, respectively, to $0.211$ and $0.411$.

In addition, for each $n=6,7,8$ and for $\sigma=0.75$, we measure the CPU time
required to process 
$L=1,000$ messages
by a single iteration of this algorithm.
These CPU times
are listed in the second column of Table~\ref{t:test23_24}.

\begin{table}[htbp]
\begin{center}
\begin{tabular}{c|c|c|c|c|c|c}
\multicolumn{1}{c|}{} & 
\multicolumn{1}{c|}{} & 
\multicolumn{5}{c}{$\sigma$} \\
\hline
$n$ & min\_RaDe1 & 0.25 & 0.5 & 0.75 & 1 & 1.25 \\
 \hline
6 & 1 &  0.587E+00 &  0.211E+00 &  0.103E+00 &  0.737E-01 &  0.588E-01 \\
 & &  0.749E+00 &  0.411E+00 &  0.268E+00 &  0.219E+00 &  0.168E+00 \\
 \hline
6 & 2 &  0.737E+00 &  0.342E+00 &  0.185E+00 &  0.128E+00 &  0.802E-01 \\
 & &  0.843E+00 &  0.549E+00 &  0.381E+00 &  0.316E+00 &  0.222E+00 \\
 \hline
6 & 3 &  0.791E+00 &  0.427E+00 &  0.267E+00 &  0.207E+00 &  0.142E+00 \\
 & &  0.879E+00 &  0.620E+00 &  0.465E+00 &  0.404E+00 &  0.337E+00 \\
 \hline
6 & 4 &  0.836E+00 &  0.478E+00 &  0.297E+00 &  0.212E+00 &  0.160E+00 \\
 & &  0.904E+00 &  0.667E+00 &  0.502E+00 &  0.426E+00 &  0.353E+00 \\
 \hline
6 & 5 &  0.858E+00 &  0.509E+00 &  0.343E+00 &  0.258E+00 &  0.198E+00 \\
 & &  0.922E+00 &  0.688E+00 &  0.545E+00 &  0.472E+00 &  0.422E+00 \\
 \hline
6 & 6 &  0.878E+00 &  0.554E+00 &  0.355E+00 &  0.296E+00 &  0.222E+00 \\
 & &  0.932E+00 &  0.719E+00 &  0.558E+00 &  0.473E+00 &  0.409E+00 \\
 \hline
6 & 7 &  0.892E+00 &  0.582E+00 &  0.374E+00 &  0.317E+00 &  0.225E+00 \\
 & &  0.940E+00 &  0.735E+00 &  0.574E+00 &  0.494E+00 &  0.404E+00 \\
 \hline
 \hline
7 & 1 &  0.678E+00 &  0.310E+00 &  0.155E+00 &  0.900E-01 &  0.622E-01 \\
 & &  0.839E+00 &  0.526E+00 &  0.334E+00 &  0.227E+00 &  0.174E+00 \\
 \hline
7 & 2 &  0.793E+00 &  0.419E+00 &  0.242E+00 &  0.155E+00 &  0.137E+00 \\
 & &  0.902E+00 &  0.631E+00 &  0.465E+00 &  0.325E+00 &  0.267E+00 \\
 \hline
7 & 3 &  0.844E+00 &  0.487E+00 &  0.293E+00 &  0.182E+00 &  0.137E+00 \\
 & &  0.929E+00 &  0.689E+00 &  0.514E+00 &  0.366E+00 &  0.324E+00 \\
 \hline
7 & 4 &  0.875E+00 &  0.541E+00 &  0.348E+00 &  0.245E+00 &  0.192E+00 \\
 & &  0.942E+00 &  0.731E+00 &  0.562E+00 &  0.444E+00 &  0.373E+00 \\
 \hline
7 & 5 &  0.903E+00 &  0.576E+00 &  0.364E+00 &  0.261E+00 &  0.199E+00 \\
 & &  0.955E+00 &  0.756E+00 &  0.575E+00 &  0.446E+00 &  0.376E+00 \\
 \hline
7 & 6 &  0.910E+00 &  0.600E+00 &  0.400E+00 &  0.273E+00 &  0.210E+00 \\
 & &  0.959E+00 &  0.771E+00 &  0.605E+00 &  0.465E+00 &  0.376E+00 \\
 \hline
7 & 7 &  0.924E+00 &  0.624E+00 &  0.426E+00 &  0.307E+00 &  0.220E+00 \\
 & &  0.964E+00 &  0.790E+00 &  0.619E+00 &  0.513E+00 &  0.412E+00 \\
 \hline
 \hline
8 & 1 &  0.757E+00 &  0.288E+00 &  0.129E+00 &  0.631E-01 &  0.367E-01 \\
 & &  0.900E+00 &  0.514E+00 &  0.296E+00 &  0.184E+00 &  0.147E+00 \\
 \hline
8 & 2 &  0.845E+00 &  0.393E+00 &  0.186E+00 &  0.105E+00 &  0.653E-01 \\
 & &  0.941E+00 &  0.617E+00 &  0.376E+00 &  0.247E+00 &  0.188E+00 \\
 \hline
8 & 3 &  0.883E+00 &  0.446E+00 &  0.233E+00 &  0.122E+00 &  0.898E-01 \\
 & &  0.959E+00 &  0.658E+00 &  0.425E+00 &  0.263E+00 &  0.216E+00 \\
 \hline
8 & 4 &  0.902E+00 &  0.475E+00 &  0.250E+00 &  0.154E+00 &  0.114E+00 \\
 & &  0.965E+00 &  0.683E+00 &  0.457E+00 &  0.302E+00 &  0.257E+00 \\
 \hline
8 & 5 &  0.910E+00 &  0.508E+00 &  0.278E+00 &  0.152E+00 &  0.102E+00 \\
 & &  0.965E+00 &  0.697E+00 &  0.475E+00 &  0.327E+00 &  0.282E+00 \\
 \hline
8 & 6 &  0.920E+00 &  0.539E+00 &  0.290E+00 &  0.173E+00 &  0.114E+00 \\
 & &  0.968E+00 &  0.719E+00 &  0.486E+00 &  0.357E+00 &  0.278E+00 \\
 \hline
8 & 7 &  0.931E+00 &  0.546E+00 &  0.300E+00 &  0.184E+00 &  0.114E+00 \\
 & &  0.974E+00 &  0.729E+00 &  0.493E+00 &  0.350E+00 &  0.294E+00 \\
\end{tabular}
\end{center}
\caption{\it Success rate of RaDe1, 
with and without supercharging ($k_1=2n^2$). 
Corresponds to Experiment 3 in Section~\ref{sec_exp_3}.
}
\label{t:test25}
\end{table}

{\bf Observations.}
The following observations can be made from Table~\ref{t:test25}
and some additional experiments by the authors.

{\bf 1.} Not surprisingly, for each $n$ and $\sigma$ the performance
of the scheme improves as the number of iterations $T$ increases 
(obviously, this
improvement comes at the cost of additional CPU time).
See, however, Remark~\ref{rem_RaDe1_stuck} below.

{\bf 2.} As expected, for each $n$, $\sigma$ and $T$, the performance
of the scheme with supercharging is better than without supercharging
(again, at the cost of additional CPU time).

{\bf 3.} The effects of supercharging are generally rather negligible
when the proportion of "correct guesses" is already high. For example,
when $n=8$, $\sigma=0.25$ and $T=7$, supercharging improves the success
rate from 93.1\% to 97.4\%. On the other hand, the improvement
might be quite noticeable when the success rate is relatively low.
For example, when $n=8$, $\sigma=0.75$ and $T=7$, supercharging
improves the success rate from 30\% to 49\%.
 
{\bf 4.} As expected, for each $n$ and $T$ the performance of the
scheme deteriorates as $\sigma$ increases (noise of a larger
standard deviation makes
the decoding task more difficult).
\begin{remark}
Additional numerical experiments 
seem to indicate that, for sufficiently large noise, 
the success rate
of RaDe1 does not approach 100\% 
as the number of iterations increases to some reasonable value
(e.g. $T = 100$), but seems to be stuck at some
intermediate value (e.g. 70\%). This phenomenon is likely to be related
to the fact that RaDe1 is not fully randomized, i.e.
the parameter $j$ from Table~\ref{t:stage_12} 
(see also Section~\ref{sec_RaDe1_informal}) is fixed (and is determined
by $H$). To get rid of this undesirable feature, in the future
we will investigate the possibility of choosing $j$ at random,
probably taking into account the values $S_{n-1}(j)$
(see \eqref{eq_skj_2} in Section~\ref{sec_RaDe1_overview}, 
and also Table~\ref{t:stage_12}).
\label{rem_RaDe1_stuck}
\end{remark}
Some additional observations can be made from Table~\ref{t:test23_24}.

{\bf 5.} A single iteration of RaDe1
(see Section~\ref{sec_RaDe1_detailed}) is faster than the nearest search
in $X$ with $k=2 \cdot n^2$ (see Section~\ref{sec_nn_x}) by
a factor of about 1.6 for $n=6$ and by a factor of
about 5 for $n=8$. Needless to say, for larger values of $k$ the
difference in CPU times is even more significant. For example, 
for $n=8$ and $k=n^5$ the nearest neighbors search in $X$
is about 800 times slower.

{\bf 6.} Since the CPU time required by a 
a nearest neighbors search in $X$ with $k=2 \cdot n^2$ is of the same
order of magnitude as that required by a single
iteration of RaDe1,
it might make sense to perform fewer iterations of the basic
scheme followed by a supercharging rather than
performing more iterations of the basic scheme without supercharging
(see also Table~\ref{t:test25}).

\subsection{Experiment 4}
\label{sec_exp_4}

In this experiment, we investigate the performance
of Randomized Decoder 2 (see Section~\ref{sec_RaDe2}).

{\bf Description.} 
This experiment is similar to that described in Section~\ref{sec_exp_3},
with the difference that instead of RaDe1 (see
Section~\ref{sec_RaDe1}) we use RaDe2 (see Section~\ref{sec_RaDe2}).
More specifically,
for each $n=6,7,8$ and each $\sigma=0.25,0.5,0.75,1,1.25$,
we proceed as in Experiment 1 from Section~\ref{sec_exp_1}. 
Then, for each received message $y_{\text{obs}}$ and each $T=1,\dots,7$,
we evaluate the estimate $x_{\text{RaDe2}}(T)$ of $x_{\text{true}}$
from $y_{\text{obs}}$
via RaDe2, with the following parameters
(see 
\eqref{eq_chi_stop}, \eqref{eq_RaDe2_calling} in Section~\ref{sec_RaDe2_full},
and also Table~\ref{t:some_params} in Section~\ref{sec_stage_3}):
\begin{align}
& \text{min\_RaDe2} = T, \nonumber \\
& \text{max\_RaDe2} = T, \nonumber \\
& \chi_{\text{stop}} = 10^{-3}, \nonumber \\
& \chi_{\text{thresh}}: \text{not applicable}
\label{eq_exp_4_params}
\end{align}
(the parameter $\chi_{\text{thresh}}$ is not applicable since
 $\text{min\_RaDe2} = \text{max\_RaDe2} = T$, e.g. the number
of basic iterations is always exactly $T$). Then, 
we define the integer $k_1$ via \eqref{eq_exp_3_k1} above,
and evaluate
the estimate $x_{\text{super}}(T)$ of $x_{\text{true}}$
from $y_{\text{obs}}$ via RaDe2
(where the parameters are defined via \eqref{eq_exp_4_params})
followed by supercharging (see Section~\ref{sec_super},
in particular \eqref{eq_super_calling}); during the supercharging
step, $k_1$ nearest neighbors are used. In other words,
$x_{\text{RaDe2}}(T)$ and $x_{\text{super}}(T)$ are obtained 
from $y_{\text{obs}}$ via
calling
\begin{align}
& 
\text{RaDe2\_search}\left( y_{\text{obs}}; x_{\text{RaDe2}}(T), r, \chi\right)
\nonumber \\
& 
\text{RaDe2\_search}\left( y_{\text{obs}}; x, r, \chi\right)
\nonumber \\
&
\text{super}\left( y_{\text{obs}}; k_1, x, r, \chi; 
x_{\text{super}}(T), r_{\text{super}}, \chi_{\text{super}} \right)
\label{eq_exp_4_calling}
\end{align}
(see \eqref{eq_RaDe2_calling}, \eqref{eq_super_calling}),
where the parameters are defined via \eqref{eq_exp_4_params}
and \eqref{eq_exp_3_k1}.

Thus, for each $n$ and $\sigma$, we have 5,000 messages $y_{\text{obs}}$
(five matrices per $n$,
$L=1,000$ messages per matrix,
as in Experiments 1,2 from Sections~\ref{sec_exp_1}, \ref{sec_exp_2},
\ref{sec_exp_3});
for each such message, we obtain 14 estimates
$x_{\text{RaDe2}}(1), \dots, x_{\text{RaDe2}}(7)$ and
$x_{\text{super}}(1), \dots, x_{\text{super}}(7)$
of $x_{\text{true}}$. For each such estimate, we evaluate
the proportion of the cases in which this estimate is correct
(provided that $x_{\text{true}} = x_{\text{best}}$, see
Section~\ref{sec_exp_1}). In other words, we compute
$\text{prop}_{\text{RaDe2}}(n,\sigma,T)$ and
$\text{prop}_{\text{super}}(n,\sigma,T)$ via the formulae
\begin{align}
\text{prop}_{\text{RaDe2}}(n,\sigma,T) =
\frac{ \# \left\{ x_{\text{best}} = x_{\text{true}} = 
                  x_{\text{RaDe2}(T)} \right\} }
     { \# \left\{ x_{\text{best}} = x_{\text{true}} \right\} }
\label{eq_exp_4_prop_RaDe2}
\end{align}
and
\begin{align}
\text{prop}_{\text{super}}(n,\sigma,T) =
\frac{ \# \left\{ x_{\text{best}} = x_{\text{true}} = 
                  x_{\text{super}(T)} \right\} }
     { \# \left\{ x_{\text{best}} = x_{\text{true}} \right\} },
\label{eq_exp_4_prop_super}
\end{align}
respectively (see also Remark~\ref{rem_exp_2} in Section~\ref{sec_exp_2}).

{\bf Tables.} The results of Experiment 4 are displayed
in Table~\ref{t:test26}, whose structure is similar
to that of Table~\ref{t:test25}. More specifically, 
in this table, the rows correspond to all pairs of $n=6,7,8$
and $T=1,\dots,7$, and the columns correspond to 
$\sigma=0.25,0.5,0.75,1,1.25$. For each $n,\sigma,T$, the
table contains two entries that appear one under the other: 
$\text{prop}_{\text{RaDe2}}(n,\sigma,T)$
(above)
and $\text{prop}_{\text{super}}(n,\sigma,T)$ (below);
see \eqref{eq_exp_4_prop_RaDe2},
\eqref{eq_exp_4_prop_super}. For example, for $n=6$, $T=1$ and $\sigma=0.5$,
these proportions are equal, respectively, to $0.533$ and $0.775$.

In addition, for each $n=6,7,8$ and for $\sigma=0.75$, we measure the CPU time
required to process 
$L=1,000$ messages
by a single iteration of this algorithm.
These CPU times
are listed in the third column of Table~\ref{t:test23_24}.

\begin{table}[htbp]
\begin{center}
\begin{tabular}{c|c|c|c|c|c|c}
\multicolumn{1}{c|}{} & 
\multicolumn{1}{c|}{} & 
\multicolumn{5}{c}{$\sigma$} \\
\hline
$n$ & min\_RaDe2 & 0.25 & 0.5 & 0.75 & 1 & 1.25 \\
 \hline
6 & 1 &  0.969E+00 &  0.752E+00 &  0.533E+00 &  0.429E+00 &  0.302E+00 \\
 & &  0.987E+00 &  0.902E+00 &  0.775E+00 &  0.696E+00 &  0.612E+00 \\
 \hline
6 & 2 &  0.989E+00 &  0.871E+00 &  0.713E+00 &  0.609E+00 &  0.529E+00 \\
 & &  0.997E+00 &  0.955E+00 &  0.870E+00 &  0.804E+00 &  0.735E+00 \\
 \hline
6 & 3 &  0.995E+00 &  0.912E+00 &  0.798E+00 &  0.697E+00 &  0.642E+00 \\
 & &  0.998E+00 &  0.965E+00 &  0.915E+00 &  0.849E+00 &  0.821E+00 \\
 \hline
6 & 4 &  0.996E+00 &  0.938E+00 &  0.841E+00 &  0.771E+00 &  0.709E+00 \\
 & &  0.998E+00 &  0.978E+00 &  0.930E+00 &  0.892E+00 &  0.866E+00 \\
 \hline
6 & 5 &  0.998E+00 &  0.952E+00 &  0.871E+00 &  0.807E+00 &  0.781E+00 \\
 & &  0.100E+01 &  0.984E+00 &  0.948E+00 &  0.926E+00 &  0.874E+00 \\
 \hline
6 & 6 &  0.998E+00 &  0.962E+00 &  0.892E+00 &  0.839E+00 &  0.807E+00 \\
 & &  0.999E+00 &  0.985E+00 &  0.960E+00 &  0.930E+00 &  0.912E+00 \\
 \hline
6 & 7 &  0.998E+00 &  0.968E+00 &  0.909E+00 &  0.860E+00 &  0.834E+00 \\
 & &  0.100E+01 &  0.988E+00 &  0.962E+00 &  0.949E+00 &  0.914E+00 \\
 \hline
 \hline
7 & 1 &  0.838E+00 &  0.536E+00 &  0.349E+00 &  0.262E+00 &  0.202E+00 \\
 & &  0.895E+00 &  0.714E+00 &  0.562E+00 &  0.448E+00 &  0.446E+00 \\
 \hline
7 & 2 &  0.918E+00 &  0.657E+00 &  0.488E+00 &  0.387E+00 &  0.334E+00 \\
 & &  0.948E+00 &  0.804E+00 &  0.673E+00 &  0.578E+00 &  0.510E+00 \\
 \hline
7 & 3 &  0.954E+00 &  0.727E+00 &  0.556E+00 &  0.440E+00 &  0.396E+00 \\
 & &  0.973E+00 &  0.854E+00 &  0.735E+00 &  0.630E+00 &  0.570E+00 \\
 \hline
7 & 4 &  0.963E+00 &  0.784E+00 &  0.598E+00 &  0.494E+00 &  0.415E+00 \\
 & &  0.978E+00 &  0.887E+00 &  0.757E+00 &  0.688E+00 &  0.624E+00 \\
 \hline
7 & 5 &  0.976E+00 &  0.802E+00 &  0.648E+00 &  0.535E+00 &  0.456E+00 \\
 & &  0.985E+00 &  0.896E+00 &  0.785E+00 &  0.690E+00 &  0.655E+00 \\
 \hline
7 & 6 &  0.981E+00 &  0.834E+00 &  0.674E+00 &  0.564E+00 &  0.503E+00 \\
 & &  0.989E+00 &  0.913E+00 &  0.818E+00 &  0.735E+00 &  0.699E+00 \\
 \hline
7 & 7 &  0.982E+00 &  0.859E+00 &  0.705E+00 &  0.619E+00 &  0.516E+00 \\
 & &  0.988E+00 &  0.934E+00 &  0.837E+00 &  0.748E+00 &  0.692E+00 \\
 \hline
 \hline
8 & 1 &  0.764E+00 &  0.370E+00 &  0.229E+00 &  0.158E+00 &  0.155E+00 \\
 & &  0.863E+00 &  0.554E+00 &  0.393E+00 &  0.281E+00 &  0.224E+00 \\
 \hline
8 & 2 &  0.888E+00 &  0.494E+00 &  0.314E+00 &  0.242E+00 &  0.200E+00 \\
 & &  0.941E+00 &  0.699E+00 &  0.506E+00 &  0.427E+00 &  0.347E+00 \\
 \hline
8 & 3 &  0.935E+00 &  0.586E+00 &  0.400E+00 &  0.298E+00 &  0.249E+00 \\
 & &  0.966E+00 &  0.760E+00 &  0.588E+00 &  0.480E+00 &  0.400E+00 \\
 \hline
8 & 4 &  0.958E+00 &  0.652E+00 &  0.420E+00 &  0.334E+00 &  0.290E+00 \\
 & &  0.980E+00 &  0.813E+00 &  0.626E+00 &  0.523E+00 &  0.465E+00 \\
 \hline
8 & 5 &  0.970E+00 &  0.704E+00 &  0.461E+00 &  0.363E+00 &  0.306E+00 \\
 & &  0.982E+00 &  0.841E+00 &  0.665E+00 &  0.568E+00 &  0.527E+00 \\
 \hline
8 & 6 &  0.977E+00 &  0.732E+00 &  0.512E+00 &  0.393E+00 &  0.302E+00 \\
 & &  0.989E+00 &  0.859E+00 &  0.697E+00 &  0.598E+00 &  0.502E+00 \\
 \hline
8 & 7 &  0.977E+00 &  0.759E+00 &  0.545E+00 &  0.417E+00 &  0.339E+00 \\
 & &  0.992E+00 &  0.873E+00 &  0.717E+00 &  0.622E+00 &  0.518E+00 \\
\end{tabular}
\end{center}
\caption{\it Success rate of RaDe2, 
with and without supercharging ($k_1=2n^2$). 
Corresponds to Experiment 4 in Section~\ref{sec_exp_4}.
}
\label{t:test26}
\end{table}

{\bf Observations.}
The following observations can be made from Table~\ref{t:test26}
and some additional experiments by the authors.
These observations are somewhat similar to those 
from Section~\ref{sec_exp_3}.

{\bf 1.} Not surprisingly, for each $n$ and $\sigma$ the performance
of the scheme improves as the number of iterations $T$ increases 
(obviously, this
improvement comes at the cost of additional CPU time).
See, however, Remark~\ref{rem_RaDe2_stuck} below.

{\bf 2.} As expected, for each $n$, $\sigma$ and $T$, the performance
of the scheme with supercharging is better than without supercharging
(again, at the cost of additional CPU time).

{\bf 3.} The effects of supercharging are generally rather negligible
when the proportion of "correct guesses" is already high. For example,
when $n=8$, $\sigma=0.25$ and $T=3$, supercharging improves the success
rate from 93.5\% to 96.6\%. On the other hand, the improvement
might be quite noticeable when the success rate is relatively low.
For example, when $n=8$, $\sigma=0.75$ and $T=2$, supercharging
improves the success rate from 31.4\% to 50.6\%.
 
{\bf 4.} As expected, for each $n$ and $T$ the performance of RaDe2
deteriorates as $\sigma$ increases (noise of a larger
standard deviation makes
the decoding task more difficult).
\begin{remark}
Additional numerical experiments 
seem to indicate that, for sufficiently large noise, 
the success rate
of RaDe2 does not approach 100\% 
as the number of iterations increases to some reasonable value
(e.g. $T = 100$), but seems to be stuck at some
intermediate value (e.g. 70\%). 
In addition, this intermediate value is larger that the related value
of RaDe1
(see Remark~\ref{rem_RaDe1_stuck} in Section~\ref{sec_exp_3}).
This phenomenon is likely to be related
to the fact that RaDe2 is not fully randomized, i.e.
the parameters $j_1,j_2$ from Table~\ref{t:stage_12} 
(see also Section~\ref{sec_RaDe2_informal}) are fixed (and are determined
by $H$). To get rid of this undesirable feature, in the future
we will investigate the possibility of choosing $j_1,j_2$ at random,
probably taking into account the values $S_{n-2}(j_1), S_{n-2}(j_2)$
(see \eqref{eq_skj_2} in Section~\ref{sec_RaDe1_overview}, 
and also Table~\ref{t:stage_12}).
\label{rem_RaDe2_stuck}
\end{remark}
Some additional observations can be made from Table~\ref{t:test23_24}.

{\bf 5.} A single iteration of RaDe2
(see Section~\ref{sec_RaDe2_detailed}) is typically somewhat 
slower than a single iteration of RaDe1
(by a factor between 2.4 and 3.1 for $n=6,7,8$
and $\sigma = 0.75$).
Obviously, this is expected from \eqref{eq_cost_RaDe1}, \eqref{eq_cost_RaDe2}.

{\bf 6.} Since the CPU time required by a 
a nearest neighbors search in $X$ with $k=2 \cdot n^2$ is of the same
order of magnitude as that required by a single
iteration of RaDe2,
it might make sense to perform fewer iterations of the basic
scheme followed by a supercharging rather than
performing more iterations of the basic scheme without supercharging
(see also Table~\ref{t:test26}).

\subsection{Experiment 5}
\label{sec_exp_5}

In this experiment, we compare the performance of the three 
decoding schemes: 
nearest neighbors search in $X$, RaDe1 and RaDe2
(see Sections~\ref{sec_nn_x}, \ref{sec_RaDe1},
\ref{sec_RaDe2}, respectively). Roughly speaking, the purpose
of this experiment is to determine empirically
how long it takes for each of the schemes to achieve
a certain success rate (for several choices of $n$
and noise deviation $\sigma$). In this sense, this experiment
can be viewed simply as a comparison of the results
of Experiments 2,3,4 above. 
For the sake of completeness, however,
we provide a detailed description below.

For each $n=6,7,8$ and each $\sigma = 0.25, 0.75, 1.25$,
we proceed as follows. First, we repeat Experiment 1 from
Section~\ref{sec_exp_1}; thus, for each of the 5,000 
randomly selected transmitted messages $x_{\text{true}}$ in $X$, 
we obtain a received message $y_{\text{obs}}$ in $\Cc^n$ and
a maximum likelihood estimate $x_{\text{best}}$ of $x_{\text{true}}$
(see Section~\ref{sec_exp_1}). As in Experiments 2,3,4 above,
to evaluate the performance of a decoding scheme we
use only those messages for which $x_{\text{best}} = x_{\text{true}}$
(see Remark~\ref{rem_exp_2} in Section~\ref{sec_exp_2}).

Then we choose, more or less arbitrarily, the desired success rate
$0 < p < 1$, and, for each of the schemes from 
Sections~\ref{sec_nn_x}, \ref{sec_RaDe1}, \ref{sec_RaDe2},
we choose the parameters in such a way that the success rate of
the scheme is roughly equal to $p$
(see e.g. \eqref{eq_prop_nn_x} in Section~\ref{sec_exp_2},
\eqref{eq_exp_3_prop_RaDe1}, \eqref{eq_exp_3_prop_super}
in Section~\ref{sec_exp_3},
\eqref{eq_exp_4_prop_RaDe2}, \eqref{eq_exp_4_prop_super}
in Section~\ref{sec_exp_4}). More specifically: 

for
the nearest neighbors search in $X$, we select
the appropriate number of nearest neighbors $k$ (see Section~\ref{sec_nn_x});

for RaDe1, we
select the number of iterations $T$, the number of nearest
neighbors $k_1$ for supercharging
(zero if no supercharging; see also \eqref{eq_exp_3_k1}), 
and use the parameters from \eqref{eq_exp_3_params};

for RaDe2,
select the number of iterations $T$, the number of nearest
neighbors $k_1$ for supercharging
(zero if no supercharging; see also \eqref{eq_exp_3_k1}), 
and use the parameters from \eqref{eq_exp_4_params}.

For each of the schemes from Sections~\ref{sec_nn_x},
\ref{sec_RaDe1}, \ref{sec_RaDe2},
we measure the proportion of correct guesses
($\text{prop}_{\text{nn}}(n,\sigma)$,
 $\text{prop}_{\text{RaDe1}}(n,\sigma)$,
 $\text{prop}_{\text{RaDe2}}(n,\sigma)$,
respectively; see 
\eqref{eq_prop_nn_x}, 
\eqref{eq_exp_3_prop_RaDe1},
\eqref{eq_exp_3_prop_super},
\eqref{eq_exp_4_prop_RaDe2},
\eqref{eq_exp_4_prop_super}). In addition, for each scheme,
we measure the CPU times
required to decode 1,000 messages
(see e.g. Table~\ref{t:test23_24}).

{\bf Tables.} The results of this experiment are displayed
in Tables~\ref{t:summary_6_025} -- \ref{t:summary_8_125}.
Each of this tables corresponds to a certain choice
of $n$ and $\sigma$, and has the following structure.
The columns correspond to the decoding schemes
from Sections~\ref{sec_nn_x}, \ref{sec_RaDe1}, \ref{sec_RaDe2},
respectively
(i.e. nearest neighbors search in $X$, RaDe1 and RaDe2).
 In the first row, we list the scheme's parameters
($k$ for the nearest neighbors search in $X$,
 number of iterations $T$ and supercharging parameter $k_1$
for each of RaDe1 and RaDe2).
The second row contains the proportion of correctly decoded
messages (provided that $x_{\text{true}} = x_{\text{best}}$),
i.e.
$\text{prop}_{\text{nn}}(n,\sigma)$,
 $\text{prop}_{\text{RaDe1}}(n,\sigma)$,
 $\text{prop}_{\text{RaDe2}}(n,\sigma)$,
respectively (by design, we expect these proportions
to be roughly the same across all columns).
The third row contains the CPU time (in seconds) required to decode
1,000 messages; these times are used to compare the performance
of the schemes (the faster the better).

For example, for $n=6$ and $\sigma=0.25$ (see Table~\ref{t:summary_6_025}),
it takes the scheme from Section~\ref{sec_nn_x} (with $k=2n^2+1$)
to achieve the success rate of 91\% in 0.075 seconds per 1,000 messages;
on the other hand, one iteration of RaDe2
(without supercharging) achieves the success rate of 97\% in
0.052 seconds per 1,000 messages.

\begin{table}[htbp]
\begin{center}
\begin{tabular}{c|c|c|c}
 & NN in $X$ & RaDe1 & RaDe2 \\
\hline
parameters &  $k=2n^2+1$  & 4 iterations, $k_1=2n$  & 1 iteration, $k_1=0$ \\
proportion & 0.91E+0  & 0.90E+0          & 0.97E+0        \\
CPU time &   0.75E-1  & 0.20E+0          & 0.52E-1        \\
\end{tabular}
\end{center}
\caption{\it Corresponds to Experiment 5 with $n=6$, $\sigma=0.25$
(see Section~\ref{sec_exp_5}).}
\label{t:summary_6_025}
\end{table}

\begin{table}[htbp]
\begin{center}
\begin{tabular}{c|c|c|c}
 & NN in $X$ & RaDe1 & RaDe2 \\
\hline
parameters &  $k=2n^2+1$  & 2 iterations, $k_1=2n$  & 2 iterations, $k_1=0$ \\
proportion & 0.90E+0  & 0.90E+0          & 0.92E+0        \\
CPU time &   0.14E+0  & 0.11E+0          & 0.10E+0        \\
\end{tabular}
\end{center}
\caption{\it Corresponds to Experiment 5 with $n=7$, $\sigma=0.25$
(see Section~\ref{sec_exp_5}).}
\label{t:summary_7_025}
\end{table}

\begin{table}[htbp]
\begin{center}
\begin{tabular}{c|c|c|c}
 & NN in $X$ & RaDe1 & RaDe2 \\
\hline
parameters &  $k=n^5+1$  & 2 iterations, $k_1=0$  & 2 iterations, $k_1=0$ \\
proportion & 0.82E+0  & 0.85E+0          & 0.88E+0        \\
CPU time &   0.45E+2  & 0.10E+0          & 0.13E+0        \\
\end{tabular}
\end{center}
\caption{\it Corresponds to Experiment 5 with $n=8$, $\sigma=0.25$
(see Section~\ref{sec_exp_5}).}
\label{t:summary_8_025}
\end{table}

\begin{table}[htbp]
\begin{center}
\begin{tabular}{c|c|c|c}
 & NN in $X$ & RaDe1 & RaDe2 \\
\hline
parameters &  $k=n^4$  & 20 iterations, $k_1=2n^2$  & 3 iterations, $k_1=2n^2$ \\
proportion & 0.72E+0  & 0.70E+0          & 0.71E+0        \\
CPU time &   0.13E+1  & 0.10E+1          & 0.27E+0        \\
\end{tabular}
\end{center}
\caption{\it Corresponds to Experiment 5 with $n=6$, $\sigma=0.75$
(see Section~\ref{sec_exp_5}).}
\label{t:summary_6_075}
\end{table}

\begin{table}[htbp]
\begin{center}
\begin{tabular}{c|c|c|c}
 & NN in $X$ & RaDe1 & RaDe2 \\
\hline
parameters &  $k=n^4$  & 9 iterations, $k_1=2n^2$  & 2 iterations, $k_1=2n^2$ \\
proportion & 0.65E+0  & 0.65E+0          & 0.67E+0        \\
CPU time &   0.27E+1  & 0.58E+0          & 0.35E+0        \\
\end{tabular}
\end{center}
\caption{\it Corresponds to Experiment 5 with $n=7$, $\sigma=0.75$
(see Section~\ref{sec_exp_5}).}
\label{t:summary_7_075}
\end{table}

\begin{table}[htbp]
\begin{center}
\begin{tabular}{c|c|c|c}
 & NN in $X$ & RaDe1 & RaDe2 \\
\hline
parameters &  $k=n^5+1$  & 20 iterations, $k_1=2n^2$  & 3 iterations, $k_1=0$ \\
proportion & 0.58E+0  & 0.57E+0          & 0.59E+0        \\
CPU time &   0.45E+2  & 0.13E+1          & 0.61E+0        \\
\end{tabular}
\end{center}
\caption{\it Corresponds to Experiment 5 with $n=8$, $\sigma=0.75$
(see Section~\ref{sec_exp_5}).}
\label{t:summary_8_075}
\end{table}

\begin{table}[htbp]
\begin{center}
\begin{tabular}{c|c|c|c}
 & NN in $X$ & RaDe1 & RaDe2 \\
\hline
parameters &  $k=n^4$  & 25 iterations, $k_1=2n^2$  & 1 iteration, $k_1=2n^2$ \\
proportion & 0.60E+0  & 0.59E+0          & 0.61E+0        \\
CPU time &   0.13E+1  & 0.12E+1          & 0.30E+0        \\
\end{tabular}
\end{center}
\caption{\it Corresponds to Experiment 5 with $n=6$, $\sigma=1.25$
(see Section~\ref{sec_exp_5}).}
\label{t:summary_6_125}
\end{table}

\begin{table}[htbp]
\begin{center}
\begin{tabular}{c|c|c|c}
 & NN in $X$ & RaDe1 & RaDe2 \\
\hline
parameters &  $k=n^5+1$  & 50 iterations, $k_1=2n^2$  & 4 iterations, $k_1=2n^2$ \\
proportion & 0.60E+0  & 0.59E+0          & 0.63E+0        \\
CPU time &   0.19E+2  & 0.26E+1          & 0.88E+0        \\
\end{tabular}
\end{center}
\caption{\it Corresponds to Experiment 5 with $n=7$, $\sigma=1.25$
(see Section~\ref{sec_exp_5}).}
\label{t:summary_7_125}
\end{table}

\begin{table}[htbp]
\begin{center}
\begin{tabular}{c|c|c|c}
 & NN in $X$ & RaDe1 & RaDe2 \\
\hline
parameters &  $k=n^5+1$  & 20 iterations, $k_1=2n^2$  & 3 iterations, $k_1=2n^2$ \\
proportion & 0.36E+0  & 0.36E+0          & 0.38E+0        \\
CPU time &   0.45E+2  & 0.13E+1          & 0.82E+0        \\
\end{tabular}
\end{center}
\caption{\it Corresponds to Experiment 5 with $n=8$, $\sigma=1.25$
(see Section~\ref{sec_exp_5}).}
\label{t:summary_8_125}
\end{table}

{\bf Observations.} The following observations
can be made from Tables~\ref{t:summary_6_025} -- \ref{t:summary_8_125}
and some additional experiments by the authors. In all of these observations,
by "better performance" we mean "takes
less CPU time to achieve a similar success rate". All the times
are in seconds per 1,000 messages.

{\bf 1.} When the noise is relatively small ($\sigma=0.25$),
the performance of RaDe1 is similar to that of  RaDe2
(moreover, even a single iteration of each scheme takes
roughly the same time; compare to Table~\ref{t:test23_24}
that corresponds to $\sigma=0.75$, and see also
\eqref{eq_cost_RaDe1}, \eqref{eq_cost_RaDe2}). For example,
for $n=8$, it takes RaDe1 0.1 seconds to achieve the
success rate of 85\%, while it takes RaDe2 0.13 seconds
to achieve the success rate of 88\%
(see Table~\ref{t:summary_8_025}).

{\bf 2.} Even for $\sigma=0.25$, the performance of the nearest neighbors
search in $X$ strongly depends on $n$ (for similar success rates).
For example, when $n=6$ or $n=7$, this scheme achieves
the success rate of about 90\% in time comparable to that
of RaDe1 and RaDe2 
(see Tables~\ref{t:summary_6_025}, \ref{t:summary_7_025};
 in these, $k=2n^2+1$). On the other hand, when $n=8$, 
the nearest neighbors search in $X$ requires already
$k=n^5$ to achieve the success rate of 82\%; subsequently,
it is slower than both RaDe1 and RaDe2 schemes by a
factor of about 40 (see Table~\ref{t:summary_8_025}).

{\bf 3.} When $\sigma=0.75$, RaDe2 outperforms
RaDe1 somewhat. More specifically, 
it achieves a slightly higher success rate
(71\% vs. 70
$n=8$) faster (by a factor of 3.7 for $n=6$, by a factor
of 1.7 for $n=7$, and by a factor of 2.1 for $n=8$).
In this case, while a single iteration of RaDe2
is slower than a single iteration of RaDe1
(see Table~\ref{t:test23_24}), the latter requires more
iterations to achieve a similar success rate
(see Tables~\ref{t:summary_6_075} -- \ref{t:summary_8_075}).

{\bf 4.} When $\sigma = 0.75$, the nearest neighbors search
in $X$ is noticeably slower than RaDe2
(when the success rates are about 70\%, 65\% and 60\% for
$n=6,7,8$, respectively). More specifically, it is slower
by factors of 4.8, 7.7 and 74 for $n=6,7,8$, respectively
(see Tables~\ref{t:summary_6_075} -- \ref{t:summary_8_075}).

{\bf 5.} When the noise is large ($\sigma = 1.25$), both
RaDe1 and RaDe2 typically perform somewhat better with supercharging
than without it. In addition, in all experiments, RaDe2
outperforms RaDe1 (by factors of about 4,3 and 1.5 with
success rates of about 60\%, 60\% and 37\% for
$n=6,7,8$, respectively).

{\bf 6.} For $\sigma = 1.25$, the algorithm from Section~\ref{sec_nn_x}
requires a relatively large number $k$ of nearest neighbors
to achieve a success rate similar to that of RaDe1 and RaDe2,
in all experiments (see
Tables~\ref{t:summary_6_075} -- \ref{t:summary_8_075}).
Consequently, the nearest neighbors search in $X$ is slower
than, say, RaDe2, by a factor of
4.3, 22 and 55 with success rates about 
 60\%, 60\% and 37\% for
$n=6,7,8$, respectively.

We conclude that the nearest neighbors search in $X$ typically
underperforms compared to both RaDe1 and RaDe2 decoding schemes
(i.e. takes more time to achieve a similar success rate);
moreover, for $n=8$ the difference can be quite noticeable.
On the other hand, RaDe2 typically outperforms RaDe1.

\section{Conclusions and Future Research}
\label{sec_future}

In this paper, we presented several schemes for decoding
of digital messages sent over a noisy multivariate Gaussian channel,
and illustrated their performance via numerical experiments.

Needless to say, there is always a gap between a prototype dealing
with a single mathematical model and the multitude of applications;
we are looking forward to collaborating with researchers and engineers
in both industry and academia to shrink this gap in 
the case of our decoding schemes.

Some future directions of our research in this area are related to
the fact that our schemes need to be extensively tested
in a
demanding industrial environment
on various choices of $C, H, \sigma,$ etc.; this will lead to 
further improvements and modifications. Other developments
are related to the fact that the decoding problem
described above admits certain variations (frequently encountered
in applications).
For example, often the coordinates of the transmitted message $x$
are not independent of each other (e.g. there is a certain
redundancy, as is in case of parity checks etc.). In other words,
the set of possible message is not all of $X$ but rather
a subset of $X$
(the reader is referred, for example, to
\cite{Alamouti}, 
\cite{Erez}, 
\cite{Hochwald}, 
\cite{Jalden}, 
\cite{Koshy}, 
\cite{Koshy2}, 
\cite{Romero}, 
for further information on the subject).
In addition, certain messages might be more likely to be transmitted
than others (based, for instance, on long-term statistics).
Obviously, our schemes will need to be modified
to take all such additional information into account.

\section{Acknowledgments}
\label{sec_ack}
The author would like to thank Linda Ness and Joe Liberty
from Applied Communication Sciences 
for numerous helpful discussions.


\end{document}